\documentclass[onefignum,onetabnum]{siamart171218}
\usepackage{jin}

\usepackage{lipsum}
\usepackage{wrapfig}
\usepackage[normalem]{ulem} \usepackage[makeroom]{cancel} \usepackage[position=b]{subcaption}
\usepackage{amsfonts}
\usepackage{graphicx}
\usepackage{epstopdf}
\usepackage{algorithmic}
\ifpdf
  \DeclareGraphicsExtensions{.eps,.pdf,.png,.jpg}
\else
  \DeclareGraphicsExtensions{.eps}
\fi

\newsiamremark{remark}{Remark}
\newsiamremark{hypothesis}{Hypothesis}
\crefname{hypothesis}{Hypothesis}{Hypotheses}
\newsiamthm{claim}{Claim}

\headers{Stability-certified reinforcement learning}{M. Jin and J. Lavaei}

\title{Stability-certified reinforcement learning: A control-theoretic perspective\thanks{{This work was supported by the ONR grants N00014-17-1-2933 and N00014-15-1-2835, DARPA grant D16AP00002, and AFOSR grant FA9550-17-1-0163.}}}

\author{Ming Jin\thanks{Department of Industrial
Engineering and Operations Research, University of California, Berkeley. Email: 
  \email{jinming@berkeley.edu}.}
\and Javad Lavaei\thanks{Department of Industrial
Engineering and Operations Research, and the Tsinghua-Berkeley
Shenzhen Institute, University of California, Berkeley. 
  Email: \email{lavaei@berkeley.edu}.}}

\usepackage{amsopn}

\makeatletter
\newcommand*{\addFileDependency}[1]{
  \typeout{(#1)}
  \@addtofilelist{#1}
  \IfFileExists{#1}{}{\typeout{No file #1.}}
}
\makeatother

\ifpdf
\hypersetup{
  pdftitle={Stability-certified reinforcement learning: A control-theoretic perspective},
  pdfauthor={Ming Jin and Javad Lavaei}
}
\fi

\begin{document}

\maketitle

\begin{abstract}
  We investigate the important problem of certifying stability of reinforcement learning policies when interconnected with nonlinear dynamical systems. We show that by regulating the input-output gradients of policies, strong guarantees of robust stability can be obtained based on a proposed semidefinite programming feasibility problem. The method is able to certify a large set of stabilizing controllers by exploiting problem-specific structures; furthermore, we analyze and establish its (non)conservatism. Empirical evaluations on two decentralized control tasks, namely multi-flight formation and power system frequency regulation, demonstrate that the reinforcement learning agents can have high performance within the stability-certified parameter space, and also exhibit stable learning behaviors in the long run.
\end{abstract}

\begin{keywords}
  Reinforcement learning, robust control, policy gradient optimization, decentralized control synthesis, safe reinforcement learning
\end{keywords}

\begin{AMS}
  	68T05, 93E35, 93D09
\end{AMS}

\section{Introduction}
\label{sec:intro}

Remarkable progress has been made in reinforcement learning (RL) using (deep) neural networks to solve complex decision-making and control problems \cite{silver2016mastering}. While RL algorithms, such as policy gradient \cite{williams1992simple,kakade2002natural,schulman2015trust}, Q-learning \cite{watkins1992q,mnih2013playing}, and actor-critic methods \cite{lillicrap2015continuous,mnih2016asynchronous} aim at optimizing control performance, the security aspect is of great importance for mission-critical systems, such as autonomous cars and power grids \cite{garcia2015comprehensive,amodei2016concrete,stoica2017berkeley}. A fundamental problem is to analyze or certify stability of the interconnected system in both RL exploration and deployment stages, which is challenging due to its dynamic and nonconvex nature \cite{garcia2015comprehensive}. 

The problem under study focuses on a general continuous-time dynamical system: 
\begin{equation}
\label{equ:system}
    \dot{x}(t)=f_t(x(t),u(t)),
\end{equation}
with the state $x(t)\in\R^{n_s}$ and the control action $u(t)\in\R^{n_a}$. In general, $f_t$ can be a time-varying and nonlinear function, but for the purpose of stability analysis, we study the important case that 
\begin{equation}
f_t(x(t))= A x(t)+Bu(t)+g_t(x(t)),
\end{equation}
where $f_t$ comprises of a linear time-invariant (LTI) component $A\in\R^{n_s\times n_s}$ that is Hurwitz (i.e., every eigenvalue of $A$ has strictly negative real part), a control matrix $B\in\R^{n_s\times n_a}$, and a slowly time-varying component $g_t$ that is allowed to be nonlinear and even uncertain.\footnote{This requirement is not difficult to meet in practice, because one can linearize any nonlinear systems around the equilibrium point to obtain a linear component and a nonlinear part.} The condition that $A$ is stable is a basic requirement, but the goal of reinforcement learning is to design a controller that optimizes some performance metric that is not necessarily related to the stability condition. For feedback control, we also allow the controller to obtain observations $y(t)=Cx(t)\in\R^{n_s}$ that are a linear function of the states, where $C\in\R^{n_s\times n_s}$ may have a sparsity pattern to account for partial observations in the context of decentralized controls \cite{bakule2008decentralized}.

\begin{figure}[h!]
  \centering
  \includegraphics[width=0.66\textwidth]{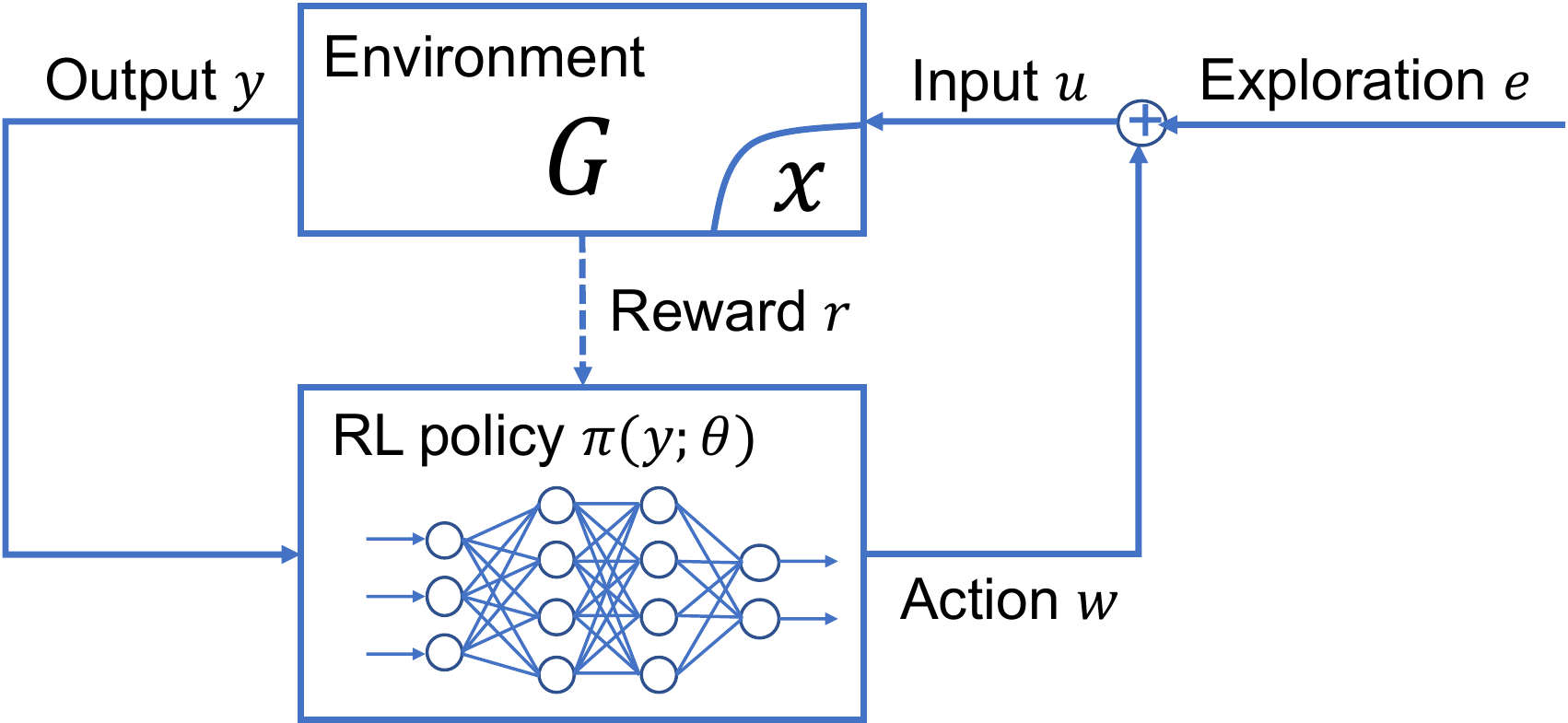}
  \caption{Overview of the interconnected system of an RL policy and the environment. The goal of RL is to maximize expected rewards through interaction and exploration.
  \label{fig:intro}}
\end{figure}

Suppose that $u(t)=\pi_t(y(t);\theta_t)+e(t)$ is a neural network  given by an RL agent (parametrized by $\theta_t$, which can be time-varying due to learning) to optimize some reward $r(x,u)$ revealed through the interaction with the environment. The exploration vector $e(t)\in\R^{n_a}$ captures the additive randomization effect during the learning phase, and is assumed to have a bounded energy over time ($\|e\|_2=\sqrt{\int |e(t)|_2^2dt}\leq\infty$). The main goal is to analyze the stability of the system with the actuation of $\pi_t$, which is typically a neural network controller, as illustrated in \cref{fig:intro}. Specifically, the stability criterion is stated using the concept of $L_2$ gain \cite{zhou1996robust,dullerud2013course}.\footnote{This stability metric is widely adopted in practice, and is closely related to bounded-input bounded-output (BIBO) stability and absolute stability (or asymptotic stability). For controllable and observable LTI systems, the equivalence can be established. }

\begin{definition}[Input-output stability]
\label{def:in-out-stable}
The $L_2$ gain of the system $G$ controlled by $\pi$ is the worst-case ratio between total output energy and total input energy:
\begin{equation}
\gamma(\G,\pi)=\sup_{u\in L_2}\frac{\|y\|_2}{\|u\|_2},
\end{equation}
where $L_2$ is the set of all square-summable signals, $\|y\|_2=\sqrt{\int |y(t)|^2_2dt}$ is the total energy over time, and $u(t)=\pi_t(y(t);\theta_t)+e(t)$ is the control input with exploration. If $\gamma(\G,\pi)$ is finite, then the interconnected system is said to have input-output stability (or finite $L_2$ gain).
\end{definition}

This study investigates the possibility of using  the gradient information of the policy $\pi_t(y(t);\theta_t)$ to obtain a stability certificate, because this information can be easily extracted in real-time and is generic enough to include a large set of performance-optimizing nonlinear controllers. Let $[n]=\{1,...,n\}$ be the set notation. By denoting
\begin{equation}
    \mathcal{P}({\xi})=\left\lbrace \pi\;\middle| \;\underline{\xi}_{ij}\leq\partial_j\pi_i(y)\leq\overline{\xi}_{ij},\forall i\in [n_a],j\in [n_s],y\in\R^{n_s} \right\rbrace
\label{equ:uncertain_set}
\end{equation}
as the set of controllers whose partial derivatives are bounded by $\underline{\xi}\in\R^{n_a\times n_s}$ and $\overline{\xi}\in\R^{n_a\times n_s}$, it is desirable to provide stability certificate as long as the RL policy remains within the above ``safety set.'' Indeed, this can be checked efficiently, as stated (informally) in the following theorem.

\begin{theorem}[Main result]
If there exist constants $\underline{\xi}$ and $\overline{\xi}$ such that the condition \eqref{equ:nltv_sdp} is feasible for the system \eqref{equ:system}, then the interconnected system has a finite $\Log_2$ gain as long as $\pi_t\in \mathcal{P}(\underline{\xi},\overline{\xi})$ for all $t\geq 0$. 
\end{theorem}
We call the constants $\underline{\xi}$ and $\overline{\xi}$ stability-certified gradient bounds for the underlying system. The above result is based on the intuition that a real-world stable controller should exhibit ``smoothness'' in the sense that small changes in the input should lead to small changes in the output. This incorporates the special case where controllers are known to have bounded Lipschitz constants (a simple strategy to calculate the Lipschitz constant of a deep neural network is suggested in \cite{szegedy2014intriguing}). To compute the gradient bounds, we borrow powerful ideas from the framework of integral quadratic constraint (in frequency domain) \cite{megretski1997system} and dissipativity theory (in time domain) \cite{willems1972dissipative} for robustness analysis. While these tools are celebrated with their non-conservatism in the robust control literature, existing characterizations of multi-input multi-output (MIMO) Lipschitz functions are insufficient. Thus, one major obstacle is to derive non-trivial bounds that could be of use in practice.

To this end, we develop a new quadratic constraint on gradient-bounded functions, which exploits the sparsity of the control architecture and the non-homogeneity of the output vector. Some key features of the stability-certified smoothness bounds are as follows: \textbf{(a)} the bounds  are inherent to the targeted real-world control task; \textbf{(b)} they can be computed efficiently by solving some semi-definite programming (SDP) problem; \textbf{(c)} they can be used to certify stability when reinforcement learning is employed in real-world control with either off-policy or on-policy learning \cite{sutton1998reinforcement}. Furthermore, the stability certification can be regarded as an $\mathcal{S}$-procedure, and we analyze its conservatism to show that it is necessary for the robustness of a surrogate system that is closely related to the original system.


The paper is organized as follows. Preliminaries on policy gradient reinforcement learning, the integrated quadratic constraint (IQC) and dissipativity frameworks are presented in \cref{sec:background}. Main results on gradient bounds for a linear or nonlinear system $G$ are presented in
\cref{sec:theory}, where we also analyze the conservatism of the certificate. The method is evaluated in \cref{sec:experiments} on two nonlinear decentralized control tasks. Conclusions are drawn in \cref{sec:conclusions}.

\section{Preliminary}
\label{sec:background}

In this section, we give an overview of the main topics relevant to this study, namely policy gradient reinforcement learning and robustness analysis based on IQC framework  and dissipativity theory.

\subsection{Reinforcement learning using policy gradient}
\label{sec:rl_polgrad}
Reinforcement learning aims at guiding an agent to perform a task as efficiently and skillfully as possible through interactions with the environment. The control task is modeled as a Markov decision process (MDP), defined by the tuple $(\mathcal{X},\mathcal{U},\mathcal{T},r,\rho)$, where $\mathcal{X}$ is the set of states $x$, $\mathcal{U}$ is a set of actions $u$, $\mathcal{T}:\mathcal{X}\times\mathcal{U}\rightarrow\mathcal{X}$ indicates the world dynamics as in \eqref{equ:system}, $r(x,u)$ is the reward at state $x$ and action $u$, and $\rho\in(0,1]$ is the factor to discount future rewards. A control strategy is defined by a policy $\pi_\theta(x)$, which can be approximated by a neural network with parameters $\theta$. For a continuous control, the actions follow a multivariate normal distribution, where $\pi_\theta(x)$ is the mean, and the standard deviation in each action dimension is set to be a diminishing number during exploration or learning, and 0 during actual deployment. With a slight abuse of notations, we use $\pi_\theta(u|x)$ to denote this normal distribution over actions, and use $x_t$ to denote $x(t)$ for simplicity. The goal of RL is to maximize the expected return:
\begin{equation}
\eta(\pi_\theta)= \underset{x_0,u_t\sim\pi_\theta(\cdot|x_t),x_{t+1}\sim\mathcal{T}(x_t,u_t)}{\mathbb{E}}\left[{\sum\nolimits_{t=0}^T} \;\;\rho^t r(x_t,u_t)\right],
\label{equ:exp_reward}
\end{equation}
where $T$ is the control horizon, and the expectation is taken over the policy, the initial state distribution  and the world dynamics. 

From a practitioner's point of view, the existing methods can be categorized into four groups based on how the optimal policy is determined: \textbf{(a)} policy gradient methods directly optimize the policy parameters $\theta$ by estimating the gradient of the expected return (e.g., REINFORCE \cite{williams1992simple}, natural policy gradient \cite{kakade2002natural}, and trust region policy optimization (TRPO)  \cite{schulman2015trust}); \textbf{(b)} {value-based algorithms} like Q-learning do not aim at optimizing the policy directly, but instead approximate the Q-value of the optimal policy for the available actions \cite{watkins1992q,mnih2013playing}; \textbf{(c)} {actor-critic algorithms} keep an estimate of the value function (critic) as well as a policy that maximizes the value function (actor) (e.g., DDPG \cite{lillicrap2015continuous} and A3C \cite{mnih2016asynchronous}); lastly, \textbf{(d)} {model-based methods} focus on the learning of the transition model for the underlying dynamics, and then use it for planning or to improve a policy (e.g., Dyna \cite{sutton1990integrated} and guided policy search \cite{levine2014learning}). We adopt an approach based on end-to-end policy gradient that combines TRPO \cite{schulman2015trust} with natural gradient \cite{kakade2002natural} and smoothness penalty (this method is very useful for RL in dynamical systems described by partial or difference equations).

\textbf{Trust region policy optimization} is a policy gradient method that constrains the step length to be within a ``trust region'' so that the local estimation of the gradient/curvature has a monotonic improvement guarantee. By manipulating the expected return $\eta(\pi)$ using the identity proposed in \cite{kakade2002approximately}, the ``surrogate objective'' $L_{\pi_{\text{old}}}(\pi)$ can be designed:
\begin{equation}
L_{\pi_{\text{old}}}(\pi)=\underset{x,u\sim\pi_{\text{old}}}{\mathbb{E}}\left[\frac{\pi(u|x)}{\pi_{\text{old}}(u|x)}\Lambda^{\pi_{\text{old}}}(x,u)\right],
\end{equation}
where the expectation is taken over the old policy $\pi_{\text{old}}$, the ratio inside the expectation is also known as the importance weight, and  $\Lambda^{\pi_{\text{old}}}(x,u)$ is the advantage function given by:
\begin{equation}
\Lambda^{\pi_{\text{old}}}(x,u)=\underset{x'\sim\mathcal{T}(x,u)}{\mathbb{E}}\left[r(x,u)+\rho V^{\pi_{\text{old}}}(x')-V^{\pi_{\text{old}}}(x)\right],
\end{equation}
where the expectation is with respect to the dynamics $x'\sim\mathcal{T}(x,u)$ (the dependence on $\theta_{\text{old}}$ is omitted), and it measures the improvement of taking action $u$ at state $x$ over the old policy in terms of the value function $V
^{\pi_{\text{old}}}$.   A bound on the difference between $L_{\pi_{\text{old}}}(\pi)$ and $\eta(\pi)$ has been derived in \cite{schulman2015trust}, which also proves a monotonic improvement result as long as the KL divergence between the new and old policies is small (i.e., the new policy stays within the trust region). In practice, the surrogate loss $L_{\pi_{\text{old}}}(\pi)$ can be estimated using trajectories sampled from $\pi_{\text{old}}$ as follows,
\begin{equation}
\widehat{L}_{\pi_{\text{old}}}(\pi)=\sum_{t}\frac{\pi(u_t|x_t)}{\pi_{\text{old}}(u_t|x_t)}\widehat{\Lambda}^{\pi_{\text{old}}}(x,u),
\end{equation}
and the averaged KL divergence over observed states $\frac{1}{T}\sum_t KL\left[\pi_{\text{old}}(\cdot|x_t),\pi(\cdot|x_t)\right]$ can be used to estimate the trust region. 

\textbf{Natural gradient} is defined by a metric based on the probability manifold induced by the KL divergence. It improves the standard gradient by making a step invariant to reparametrization of the parameter coordinates \cite{amari1998natural}:
\begin{equation}
\theta_{t+1}\leftarrow \theta_t-\lambda H_\theta^{-1}{\zeta}_t,
\end{equation}
where ${\zeta}_t$ is the standard gradient, $H_\theta=\frac{1}{T}\sum_t \left(\frac{\partial}{\partial\theta}\pi_{\theta}(\log u_t|x_t)\right)\left(\frac{\partial}{\partial\theta}\log\pi_{\theta}(u_t|x_t)\right)^\top$ is the Fisher information matrix estimated with the trajectory data, and $\lambda$ is the step size. In practice, when the number of parameters is large, conjugate gradient is employed to estimate the term $H_\theta^{-1}{\zeta}_t$ without requiring any matrix inversion. Since the Fisher information matrix coincides with the second-order approximation of the KL divergence, one can perform a back-tracking line search on the step size $\lambda$ to ensure that the updated policy stays within the trust region.

\textbf{Smoothness penalty} is introduced in this study to empirically improve learning performance on physical dynamical systems. Specifically, we propose to use 
\begin{equation}
L_{\text{explore}}=\sum\nolimits_{t=1}^T\|u_{t-1}-\pi_\theta(x_t)\|^2
\end{equation}
as a regularization term to induce consistency during exploration. The intuition is that since the change in states between two consecutive time steps is often small, it is desirable to ensure small changes in output actions. This is closely related to another penalty term that has been used in \cite{drucker1992improving}, which is termed ``double backpropagation'', and recently rediscovered in \cite{ororbia2017unifying,gulrajani2017improved}:
\begin{equation}
L_{\text{smooth}}=\sum\nolimits_{t=1}^T\left\|\frac{\partial}{\partial\theta}\pi_\theta(x_t)\right\|^2,
\end{equation}
which penalizes the gradient of the policy along the trajectories. Since bounded gradients lead to bounded Lipshitz constant, these penalties will induce smooth neural network functions, which is essential to ensure generalizability and, as we will show, stability. In addition, we incorporate a hard threshold (HT) approach that rescales the weight matrices at each layer by $(l^\circ/l(\pi_\theta))^{1/n_L}$ if $l(\pi_\theta)>l^\circ$, where $l(\pi_\theta)$ is the Lipschitz constant of the neural network $\pi_\theta$, $n_L$ is the number of layers of the neural network and $l^\circ$ is the certified Lipschitz constant. This ensures that the Lipschitz constant of the RL policy remains bounded by $l^\circ$. 

In summary, our policy gradient is based on the weighted objective:
\begin{equation}
L_{\text{pol}}(\pi_\theta)=\widehat{L}_{\pi_{\text{old}}}(\pi_\theta)+w_1L_{\text{explore}}(\pi_\theta)+w_2L_{\text{smooth}}(\pi_\theta),
\label{equ:pol_obj}
\end{equation}
where the penalty coefficients $w_1$ and $w_2$ are selected such that the scales of the corresponding terms are about $[0.01,0.05]$ of the surrogate loss value $\widehat{L}_{\pi_{\text{old}}}(\pi_\theta)$.
In each round, a set of trajectories are collected using $\pi_{\text{old}}$, which are used to estimate the gradient $\frac{\partial}{\partial\theta} L_{\text{pol}}(\pi_\theta)$ and the Fisher information matrix $H_\theta$; a backtracking line search on the step size is then conducted to ensure that the updated policy stays within the trust region. This learning procedure is known as {on-policy learning} \cite{sutton1998reinforcement}.

\subsection{Overview of IQC framework}

The IQC theory is celebrated for systematic and efficient stability analysis of a large class of uncertain, dynamic, and interconnected systems \cite{megretski1997system}. It unifies and extends classical passitivity-based multiplier theory, and has close connections to dissipativity theory in the time domain \cite{seiler2015stability}. 

To state the IQC framework, some terminologies are necessary. We define the space $L_2^n[0,\infty)=\{x:\int_{t=0}^\infty |x(t)|_2^2dt<\infty\}$ for signals supported on $t\geq 0$, where $n$ denotes the spatial dimension of $x(t)$, and the extended space $L_{2e}^n[0,\infty)=\{
x:\int_{t=0}^T|x(t)|_2^2dt<\infty,\forall T\;\geq 0\}$ (we will use $L_2$ and $L_{2e}$ if it is not necessary to specify the dimension and signal support), where we use $x$ to denote the signal in general and $x(t)$ to denote its value at time $t$. For a vector or matrix, we use superscript $*$ to denote its conjugate transpose. An operator is {causal} if the current output does not depend on future inputs. It is {bounded} if it has a finite $\Log_2$ gain. Let $\Phi:\cH\rightarrow\cH$ be a bounded linear operator on a Hilbert space. Then, its Hilbert adjoint is the operator $\Phi^*:\cH\rightarrow\cH$ such that $\left\langle\Phi x,y\right\rangle=\left\langle x,\Phi^*y\right\rangle$ for all $x,y\in\cH$, where $\left\langle\cdot,\cdot\right\rangle$ denotes the inner product. It is \emph{self-adjoint} if $\Phi=\Phi^*$. 

Consider the system (see also \cref{fig:intro})
\begin{align}
\by&=\G(\bu)\label{equ:sys1}\\
\bu&=\bDelta(\by)+\be\label{equ:sys2},
\end{align}
where $\G$ is the transfer function of a causal and bounded LTI system (i.e., it maps input $\bu\in\Log^{n_a}$ to output $\by\in\Log^{n_o}$ through the internal state dynamics $\dot{x} = Ax(t)+Bu(t)$), $\be\in\Log^{n_a}$ is the disturbance, and $\bDelta:\Log^{n_o}\rightarrow\Log^{n_a}$ is a bounded and causal function that is used to represent uncertainties in the system. IQC provides a framework to treat uncertainties such as nonlinear dynamics, model approximation and identification errors, time-varying parameters and disturbance noise, by using their {input-output characterizations}.

\begin{definition}[Integral quadratic constraints]
Consider the signals $\bw\in\Log_2$ and $\by\in\Log_2$ associated with Fourier transforms $\hat{\bw}$ and $\hat{\by}$, and $\bw=\bDelta(\by)$, where $\bDelta$ is a bounded and causal operator. We present both the frequency- and time-domain IQC definitions:
\begin{enumerate}
\item[(a)] (Frequency domain) Let ${\bPi}$ be a bounded and self-adjoint operator. Then, $\bDelta$ is said to satisfy the IQC defined by ${\bPi}$ (i.e., $\bDelta\in\text{IQC}({\bPi})$) if:
\begin{equation}
\sigma_{{\bPi}}(\hat{\by},\hat{\bw})=\int_{-\infty}^\infty\begin{bmatrix}
\hat{y}(j\omega)\\
\hat{w}(j\omega) 
\end{bmatrix}^*\Pi(j\omega)\begin{bmatrix}
\hat{y}(j\omega)\\
\hat{w}(j\omega) 
\end{bmatrix}d\omega\geq 0.
\end{equation}
\item[(b)] (Time domain) Let $(\bPsi,M)$ be any factorization of $\bPi=\bPsi^* M\bPsi$ such that $\bPsi$ is stable and $M=M^\top$. Then, $\bDelta$ is said to satisfy the \emph{hard IQC} defined by $(\bPsi,M)$ (i.e., $\bDelta\in\text{IQC}(\bPsi,M)$) if:
\begin{equation}
\int_0^T z(t)^\top M z(t)dt\geq 0,
\qquad\forall\;T\geq 0,
\end{equation}
where $\bz=\bPsi\begin{bmatrix}
\by\\
\bw
\end{bmatrix}$ is the filtered output given by the stable operator $\bPsi$. If instead of requiring nonnegativity at each time $T$, the nonnegativity is considered only when $T\rightarrow\infty$, then the corresponding condition is called \emph{soft IQC}.
\end{enumerate}
\label{def:iqc}
\end{definition}

As established in \cite{seiler2015stability}, the time- and frequency-domain IQC definitions are equivalent if there exists  $\bPi=\bPsi^*M\bPsi$ as a spectral factorization of $\bPi$ with $M=\begin{bmatrix}
I&0\\
0&-I
\end{bmatrix}$ such that $\bPsi$ and $\bPsi^{-1}$ are stable.

\begin{example}[Sector IQC]
A single-input single-output uncertainty $\Delta:\R\rightarrow\R$ is called ``sector bounded'' between $[\alpha,\beta]$ if $\alpha y(t)\leq \Delta(y(t))\leq\beta y(t)$, for all $ y\in\R$ and $t\geq 0$. It thus satisfies the sector IQC$(\bPsi,M)$ with $\bPsi=I$ and $M=\begin{bmatrix}
-2\alpha\beta&\alpha+\beta\\
\alpha+\beta&-2
\end{bmatrix}$. It also satisfies IQC$(\bPi)$ with $\bPi=M$ defined above.
\end{example}

\begin{example}[$\Log_2$ gain bound]
A MIMO uncertainty $\Delta:\R^n\rightarrow\R^m$ has the $\Log_2$ gain $\gamma$ if $\int_0^\infty\|w(t)\|^2dt\leq \gamma^2\int_0^\infty\|y(t)\|^2dt$, where $w(t)=\Delta(y(t))$. Thus, it satisfies IQC$(\bPsi,M)$ with $\bPsi=I_{n+m}$ and $M=\begin{bmatrix}
\lambda\gamma^2 I_n&0\\
0&-\lambda I_m
\end{bmatrix}$, where $\lambda>0$. It also satisfies IQC$(\bPi)$ with $\bPi=M$ defined above. This can be used to characterize  nonlinear operators with fast time-varying parameters.
\end{example}

Before stating a stability result, we define the system \eqref{equ:sys1}--\eqref{equ:sys2} (see \cref{fig:intro}) to be well-posed if for any $\be\in\Log_{2e}$, there exists a solution $\bu\in\Log_{2e}$, which depends causally on $\be$. A main IQC result for stability is stated below:

\begin{theorem}[\cite{megretski1997system}]
Consider the interconnected system \eqref{equ:sys1}--\eqref{equ:sys2}. Assume that: \textbf{(i)} the interconnected system $(\G,\tau\bDelta)$ is well posed for all $\tau\in[0,1]$; \textbf{(ii)} $\tau\bDelta\in\text{IQC}(\bPi)$ for $\tau\in[0,1]$; and \textbf{(iii)} there exists $\epsilon>0$ such that
\begin{equation}
\begin{bmatrix}
\hG(j\omega)\\
I(j\omega)
\end{bmatrix}^*\Pi(j\omega)\begin{bmatrix}
\hG(j\omega)\\
I(j\omega)
\end{bmatrix}\leq -\epsilon I,\;\qquad\forall\;\omega\in[0,\infty).
\label{equ:iqc_stable}
\end{equation}
Then, the system \eqref{equ:sys1}--\eqref{equ:sys2} is input-output stable (i.e., finite $\Log_2$ gain).
\end{theorem}

The above theorem requires three technical conditions. The well-posedness condition is a generic property for any acceptable model of a physical system. The second condition is implied if $\bPi=\begin{bmatrix}
\bPi_{11}&\bPi_{12}\\
\bPi_{12}^*&\bPi_{22}
\end{bmatrix}$ has the properties $\bPi_{11}\succeq 0$ and $\bPi_{22}\preceq 0$. The third condition is central, and it requires checking the feasibility at every frequency, which represents a main obstacle. As discussed in Section \cref{sec:compute}, this condition can be equivalently represented as a linear matrix inequality (LMI) using the Kalman-Yakubovich-Popov (KYP) lemma. In general, the more IQCs exist for the uncertainty, the better characterization can be obtained. If $\bDelta\in\text{IQC}(\bPi_k)$, $k\in[n_K]$, where $n_K$ is the number of IQCs satisfied by $\bDelta$,  then it is easy to show that $\bDelta\in\text{IQC}(\sum_{k=1}^{n_K}\tau_k\bPi_k)$, where $\tau_k\geq 0$; thus, the stability test \eqref{equ:iqc_stable} becomes a convex program, i.e., to find $\tau_k\geq 0$ such that:
\begin{equation}
\begin{bmatrix}
\hG(j\omega)\\
I(j\omega)
\end{bmatrix}^*\left(\sum_{k=1}^{n_K}\tau_k\Pi_k(j\omega)\right)\begin{bmatrix}
\hG(j\omega)\\
I(j\omega)
\end{bmatrix}\leq -\epsilon I,\;\forall\;\omega\in[0,\infty).
\label{equ:iqc_stable2}
\end{equation}

The counterpart for the frequency-domain stability condition in the time-domain can be stated using a standard dissipation argument \cite{seiler2015stability}.

\subsection{Related work}

To close this section, we summarize some connections to existing literature.  This work is closely related to the body of works on \emph{safe reinforcement learning}, defined as the process of learning policies that maximize performance in problems where safety is required during the learning and/or deployment \cite{garcia2015comprehensive}. A detailed literature review can be found in \cite{garcia2015comprehensive}, which has categorized two main approaches by modifying: \textbf{(1)} the optimality condition with a safety factor, and \textbf{(2)} the exploration process to incorporate external knowledge or risk metrics. Risk-aversion can be specified in the reward function, for example, by defining risk as the probability of reaching a set of unknown states in a discrete Markov decision process setting \cite{coraluppi1999risk,geibel2005risk}. Robust MDP is designed to maximize rewards while safely exploring the discrete state space \cite{moldovan2012safe,wiesemann2013robust}. For continuous states and actions, robust model predictive control can be employed to ensure robustness and safety constraints for the learned model with bounded errrors \cite{aswani2013provably}. These methods require an accurate or estimated models for policy learning. Recently, model-free policy optimization has been successfully demonstrated in real-world tasks such as robotics, business management, smart grid and transportation \cite{li2017deep}. Safety requirement is high in these settings. Existing approaches are based on constraint satisfaction that holds with high probability \cite{sui2015safe,achiam2017constrained}. 

The present analysis tackles the safe reinforcement learning problem from a robust control perspective, which is aimed at providing theoretical guarantees for stability \cite{zhou1996robust}.  Lyapunov functions are widely used to analyze and verify stability when the system and its controller are known \cite{perkins2002lyapunov,bobiti2016sampling}. For nonlinear systems without global convergence guarantees, region of convergence is often estimated, where any state trajectory that starts within this region stays within the region for all times and converges to a target state eventually \cite{khalil1996noninear}. For example, recently, \cite{berkenkamp2017safe}  has proposed a learning-based Lyapunov stability verification for physical systems, whose dynamics are sequentially estimated by Gaussian processes. In the same vein, \cite{akametalu2014reachability} has employed reachability analysis to construct safe regions in the state space by solving a partial differential equation. The main challenge of these methods is to find a suitable non-conservative Lyapunov function to conduct the analysis. 

The IQC framework proposed in \cite{megretski1997system} has been widely used to analyze the stability of large-scale complex systems such as aircraft control \cite{fry2017iqc}. The main advantages of IQC are its computational efficiency, non-conservatism, and unified treatment of a variety of nonlinearities and uncertainties. It has also been employed to analyze the stability of small-sized neural networks in reinforcement learning \cite{kretchmara2001robust,anderson2007robust}; however, in their analysis, the exact coefficients of the neural network need to be known a priori for the static stability analysis, and a region of safe coefficients needs to be calculated at each iteration for the dynamic stability analysis. This is computationally intensive, and it quickly becomes intractable when the neural network size grows. On the contrary,  because the present analysis is based on {a broad characterization of } control functions with bounded gradients, it does not need to access the coefficients of the neural network (or any forms of the controller). In general, robust analysis using advanced methods such as structured singular value \cite{packard1993complex} or IQC can be conservative. There are only few cases where the necessity conditions can be established, such as when the uncertain operator has a block diagonal structure of bounded singular values \cite{dullerud2013course}, but this set of uncertainties is much smaller than the set of performance-oriented controllers learned by RL. To this end, we are able to reduce conservatism of the results by introducing more informative quadratic constraints for those controllers, and analyze the necessity of the certificate criteria. This significantly extends the possibilities of stability-certified reinforcement learning to large and deep neural networks in nonlinear large-scale real-world systems, whose stability is otherwise impossible to be certified using existing approaches.

\section{Main results}
\label{sec:theory}
This section will introduce a set of quadratic constraints on gradient-bounded functions, describe the computation of a smoothness margin for linear (\cref{thm:lti_stable}) and nonlinear systems (\cref{thm:nltv_stable}). Furthermore, we examine the conservatism of the certificate condition in \cref{thm:lti_stable} for linear systems.

\subsection{Quadratic constraints on gradient-bounded functions}
\label{sec:constraint}

The starting point of this analysis is a less conservative constraint on general vector-valued functions. We start by recalling the definition of a Lipschitz continuous function:

\begin{definition}[Lipschitz continuous function]
We define both the local and global versions of the Lipschitz continuity for a function $f:\R^n\rightarrow\R^m$:
\begin{enumerate}
\item[(a)] The function $f$ is \emph{locally Lipschitz continuous} on the open subset $\B$ if there exists a constant $\xi>0$ (i.e., Lipschitz constant of $f$ on $\B$) such that
\begin{equation}
|f(x)-f(y)|\leq \xi|x-y|,\qquad \forall\;x,y\in\B.
\label{equ:lip_loc}
\end{equation}
\item[(b)] If $f$ is Lipschitz continuous on $\R^n$ with a constant $\xi$ (i.e., $\B=\R^n$ in \eqref{equ:lip_loc}), then $f$ is called \emph{globally Lipschitz continuous} with the Lipschitz constant $\xi$.
\end{enumerate}
\end{definition}
Lipschitz continuity implies uniform continuity. The above definition also establishes a connection between locally and globally Lipschitz continuity. The norm $|\cdot|$ in the definition can be any norm, but the Lipschitz constant depends on the particular choice of the norm. Unless otherwise stated, we use the Euclidean norm in our analysis.  

To explore some useful properties of Lipschitz continuity, consider a scalar-valued function (i.e., $m=1$). Let $h^{(j)}_{xy}=\begin{bmatrix}
y_1,y_2,\ldots,y_j,x_{j+1},\ldots,x_n
\end{bmatrix}^\top\in\R^n$  denote a \emph{hybrid vector} between $x$ and $y$, with $h^{(0)}_{xy}=x$ and $h^{(n)}_{xy}=y$. Then, local Lipschitz continuity of $f:\R^n\rightarrow\R$ on $\B$ implies that 
\begin{equation}
\frac{|f(h^{(j)}_{xy})-f(h^{(j-1)}_{xy})|}{|x_j-y_j|}\leq \xi, \qquad\forall\;x,y\in\B,x_j\neq y_j, j\in[n].
\end{equation}
If we were to assume that $f$ is differentiable, then it follows that its (partial) derivative is bounded by the Lipschitz constant. For a vector-valued function $f=\begin{bmatrix}
f_1,\ldots,f_m
\end{bmatrix}^\top$ that is $\xi$-Lipschitz, it is necessary that every component $f_i$ be $\xi$-Lipschitz. In general, every continuously differentiable function is locally Lipschitz, but the reverse is not true, since the definition of Lipschitz continuity does not require differentiability. Indeed, by the Rademacher's theorem, if $f$ is locally Lipschitz on $\B$, then it is differentiable at \emph{almost} every point in $\B$ \cite{clarke1990optimization}.

%


For the purpose of stability analysis, we can express \eqref{equ:lip_loc} as a point-wise quadratic constraint:
\begin{equation}
\begin{bmatrix}
x-y\\
f(x)-f(y)
\end{bmatrix}^\top\begin{bmatrix}
\xi^2I_n&0\\
0&-I_m
\end{bmatrix}\begin{bmatrix}
x-y\\
f(x)-f(y)
\end{bmatrix}\geq 0, \qquad\forall\;x,y\in\B.
\label{equ:ineq_const_1}
\end{equation}
The above constraint, nevertheless, can be sometimes too conservative, because it does not explore the structure of a given problems. To elaborate on this, consider the function $f:\R^2\rightarrow\R^2$ defined as
\begin{equation}
f(x_1,x_2)=\begin{bmatrix}
\tanh(0.5x_1)-ax_1,\sin(x_2)
\end{bmatrix}^\top,
\label{equ:f_exp}
\end{equation}
where $x_1,x_2\in\R$ and $|a|\leq 0.1$ is a deterministic but unknown parameter with a bounded magnitude. Clearly, to satisfy \eqref{equ:lip_loc} on $\R^2$ for all possible tuples $(a,x_1,x_2)$, we need to choose $\xi\geq 1$ (i.e., the function has the Lipshitz constant 1). However, this characterization is too general in this case, because it ignores the \emph{non-homogeneity} of $f_1$ and $f_2$, as well as the \emph{sparsity} of the problem representation. Indeed, $f_1$ only depends on $x_1$ with its slope restricted to $[-0.1,0.6]$ for all possible $|a|\leq 0.1$, and $f_2$ only depends on $x_2$ with its slope restricted to $[-1,1]$. In the context of controller design, the non-homogeneity of control outputs often arises from physical constraints and domain knowledge, and the sparsity of control architecture is inherent in scenarios with distributed local information. To explicitly address these requirements, we state the following quadratic constraint.

%
%

\begin{lemma}{}
For a vector-valued function $f:\R^n\rightarrow\R^m$ that is differentiable with bounded partial derivatives on $\B$ (i.e., $\underline{\xi}_{ij}\leq\partial_jf_i(x)\leq\overline{\xi}_{ij}$ for all $x\in\B$), the following quadratic constraint is satisfied for all $\lambda_{ij}\geq 0$, $i\in[m]$, $j\in[n]$, and $x,y\in\B$:
\begin{equation}
\begin{bmatrix}
x-y\\
q(x,y)
\end{bmatrix}^\top M(\lambda;\xi)\begin{bmatrix}
x-y\\
q(x,y)
\end{bmatrix}\geq 0,\label{equ:lip_con}
\end{equation}
where $M(\lambda;\xi)$ is given by
\begin{equation}
    \begin{bmatrix}
\mathrm{diag}\left(\left\lbrace \sum_{i}\lambda_{ij}(\overline{c}^2_{ij}-c^2_{ij})\right\rbrace\right)&U(\{\lambda_{ij},c_{ij}\})^\top\\
U(\{\lambda_{ij},c_{ij}\})&\mathrm{diag}\left(\{-\lambda_{ij}\}\right)
\end{bmatrix},
\label{equ:define_M}
\end{equation}
where $\mathrm{diag}(x)$ denotes a diagonal matrix with diagonal entries specified by $x$, and $q(x,y)=\begin{bmatrix}
q_{11},\ldots, q_{1n},\ldots,q_{m1},\ldots,q_{mn}
\end{bmatrix}^\top$ is determined by $x$ and $y$, $\{-\lambda_{ij}\}$ is a set of non-negative multipliers that follow the same index order as $q$, $U(\{\lambda_{ij},c_{ij}\})=\begin{bmatrix}
\mathrm{diag}\left(\{-\lambda_{1j}c_{1j}\}\right)&\cdots&\mathrm{diag}\left(\{-\lambda_{mj}c_{mj}\}\right)
\end{bmatrix}\in\R^{n\times mn}$, $c_{ij}= \frac{1}{2}\left(\underline{\xi}_{ij}+\overline{\xi}_{ij}\right)$, $\overline{c}_{ij}=\overline{\xi}_{ij}-c_{ij}$, and $q$ is related to the output of $f$ by the constraint:
\begin{equation}
f(x)-f(y)=\begin{bmatrix}
I_{m}\otimes 1_{1\times n}
\end{bmatrix}q=Wq,
\end{equation}
where $\otimes$ denotes the Kronecker product.
\label{lem:quad_const}
\end{lemma}
\begin{proof}
For a vector-valued function $f:\R^n\rightarrow\R^m$ that is differentiable with bounded partial derivatives on $\B$ (i.e., $\underline{\xi}_{ij}\leq\partial_jf_i(x)\leq\overline{\xi}_{ij}$ for all $x\in\B$), there exist functions $\delta_{ij}:\R^n\times \R^n\rightarrow\R$ bounded by $\underline{\xi}_{ij}\leq\delta_{ij}(x,y)\leq\overline{\xi}_{ij}$ for all $i\in[m]$ and $j\in[n]$ such that
\begin{equation}
f(x)-f(y)=\begin{bmatrix}
\sum_{j=1}^n\delta_{1j}(x,y)(x_j-y_j)\\
\vdots\\
\sum_{j=1}^n\delta_{mj}(x,y)(x_j-y_j)
\end{bmatrix}.
\end{equation}
By defining $q_{ij} = \delta_{ij}(x,y)(x_j-y_j)$, since $\left(\delta_{ij}(x,y)-c_{ij}\right)^2\leq \overline{c}_{ij}^2$, it follows that
\begin{equation}
\begin{bmatrix}
x_j-y_j\\
q_{ij}
\end{bmatrix}^\top\begin{bmatrix}
\overline{c}_{ij}^2-c_{ij}^2&c_{ij}\\
c_{ij}&-1
\end{bmatrix}\begin{bmatrix}
\star
\end{bmatrix}\geq 0.
\end{equation}
The result follows by introducing nonnegative multipliers $\lambda_{ij}\geq 0$, and the fact that $f_i(x)-f_i(y)=\sum_{j=1}^mq_{ij}$.
\end{proof}
This above bound is a direct consequence of standard tools in real analysis \cite{zemouche2013lmi}. To understand this result, it can be observed that  \eqref{equ:lip_con} is equivalent to:
\begin{equation}
\sum_{i,j} \lambda_{ij}\Big((\overline{c}_{ij}^2-c_{ij}^2) (x_j-y_j)^2+2c_{ij}q_{ij}(x_j-y_j)-q_{ij}^2\Big)\geq 0, \qquad\forall \lambda_{ij}\geq 0,
\label{equ:lip_con_ineq}
\end{equation}
with $f_i(x)-f_i(y) =\sum_{j=1}^n q_{ij}$, where $q_{ij}$ depends on $x$ and $y$. Since \eqref{equ:lip_con_ineq} holds for all $\lambda_{ij}\geq 0$, it is equivalent to the condition that $(\overline{c}_{ij}^2-c_{ij}^2) (x_j-y_j)^2+2c_{ij}q_{ij}(x_j-y_j)-q_{ij}^2\geq 0$ for all $i\in[m]$ and $j\in[n]$, which is a direct result of the bounds imposed on the partial derivatives of $f_i$. To illustrate its usage, let us apply the constraint to characterize the example function \eqref{equ:f_exp}, where $\underline{\xi}_{11}=-0.1,\overline{\xi}_{11}=0.6,\underline{\xi}_{22}=-1,\overline{\xi}_{22}=1$, and all the other bounds ($\underline{\xi}_{12},\overline{\xi}_{12},\underline{\xi}_{21},\overline{\xi}_{21}$) are zero. This clearly yields a more informative constraint than merely relying on the Lipschitz constraint \eqref{equ:ineq_const_1}. In fact, for a differentiable $l$-Lipschitz function, we have $\overline{\xi}_{ij}=-\underline{\xi}_{ij}=l$, and by limiting the choice of $\lambda_{ij}=\begin{cases}\lambda&\text{if }i=1\\
0&\text{if }i\neq 1\end{cases}$, \eqref{equ:lip_con_ineq} is reduced to \eqref{equ:ineq_const_1}. However, as illustrated in this example, the quadratic constraint in Lemma \ref{lem:quad_const} can incorporate richer information about the structure of the problem; therefore, it often gives rise to non-trivial stability bounds in practice.

The constraint introduced above is not a classical IQC, since it involves an intermediate variable $q$ that relates to the output $f$ through a set of linear equalities. For stability analysis, let $y=x^*\in\B$ be the equilibrium point, and without loss of generality, assume that $x^*=0$ and $f(x^*)=0$. Then, one can define the quadratic functions
\begin{equation*}
    \phi_{ij}(x,q)=(\overline{c}_{ij}^2-c_{ij}^2) x_j^2+2c_{ij}q_{ij}x_j-q_{ij}^2,
\end{equation*}
and the condition \eqref{equ:lip_con} can be written as
\begin{equation}
    \sum_{ij} \lambda_{ij}\phi_{ij}(x,q)\geq 0,\qquad \forall \lambda_{ij}\geq 0,
\end{equation}
which can be used to characterize the set of $(x,q)$ associated with the function $f$, as we will discuss in \cref{sec:necessity}.

To simplify the mathematical treatment, we have focused on differentiable functions in \cref{lem:quad_const}; nevertheless, the analysis can be extended to non-differentiable but continuous functions (e.g., the ReLU function $\max\{0,x\}$) using the notion of generalized gradient \cite[Chap. 2]{clarke1990optimization}. In brief, by re-assigning the bounds on partial derivatives to uniform bounds on the set of generalized partial derivatives, the constraint \eqref{equ:lip_con} can be directly applied.

In relation to the existing IQCs, this constraint has wider applications for the characterization of gradient-bounded functions. The Zames-Falb IQC introduced in \cite{zames1968stability} has been widely used for single-input single-output (SISO) functions $f:\R\rightarrow\R$, but it requires the function to be monotone with the slope restricted to $[\alpha,\beta]$ with $\alpha\geq 0$, i.e., $0\leq \alpha\leq\frac{f(x)-f(y)}{x-y}\leq \beta$ whenever $x\neq y$. The MIMO extension holds true only if the nonlinear function $f:\R^n\rightarrow\R^n$ is restricted to be the {gradient of a convex real-valued function} \cite{safonov2000zames,heath2005zames}. As for the sector IQC, the scalar version can not be used (because it requires $f_i(x)=0$ whenever there exists $j\in[n]$ such that $x_j=0$, which is extremely restrictive), and the vector version is in fact \eqref{equ:ineq_const_1}. 	In contrast, the quadratic constraint in \cref{lem:quad_const} can be applied to non-monotone, vector-valued Lipschitz functions.

\subsection{Computation of the smoothness margin}
\label{sec:compute}
With the newly developed quadratic constraint in place, this subsection explains the computation for a smoothness margin of an LTI system $\G$, whose state-space representation is given by:
\begin{equation}
\begin{cases}
\dot{x}_G&\!\!\!\!=Ax_G+Bu\\
w&\!\!\!\!=\pi(x_G)\\
u&\!\!\!\!=e+w
\end{cases}
\label{equ:lti}
\end{equation}
where $x_G\in\R^{n_s}$ is the state (the dependence on $t$ is omitted for simplicity). The system is assumed to be stable, i.e., $A$ is Hurwitz. We can connect this linear system in {feedback} with a controller $\pi:\R^{n_s}\rightarrow\R^{n_a}$. The signal $e\in\R^{n_a}$ is the exploration vector introduced in reinforcement learning, and $w\in\R^{n_a}$ is the policy action. We are interested in certifying the set of gradient bounds $\xi\in\R^{n_s\times n_a}$ of $\pi$ such that the interconnected system is input-output stable at all time $T\geq 0$, i.e.,
\begin{equation}
\int_0^T\left|y(t)\right|^2dt\leq \gamma^2\int_0^T\left|e(t)\right|^2dt,
\label{equ:lti_stable}
\end{equation}
where $\gamma$ is a finite upper bound for the $\Log_2$ gain. Let $A\succeq B$ or $A\succ B$ denote that $A-B$ is positive semidefinite or positive definite, respectively. To this end, define the $\text{SDP}(P,\lambda,\gamma,\xi)$ as follows:
\begin{equation}
\text{SDP}(P,\lambda,\gamma,\xi):\begin{bmatrix}
O(P,\lambda,\xi) &S(P)\\
      S(P)^\top & -\gamma I
\end{bmatrix}\prec 0,
\label{equ:sdp_lti}
\end{equation}
where $P=P^\top\succeq 0$ and
\begin{align*}
O(P,\lambda,\xi) &=\begin{bmatrix}
A^\top P+PA&PBW\\
W^\top B^\top P&0
\end{bmatrix}+\frac{1}{\gamma}\begin{bmatrix}
I&0\\
0&0
\end{bmatrix}+M(\lambda;\xi),\;\;
S(P)=\begin{bmatrix}
PB\\
0
\end{bmatrix},
\end{align*}
where $M(\lambda;\xi)$ is defined in \eqref{equ:define_M}.  We will show next that the stability of the interconnected system can be certified using linear matrix inequalities.

\begin{theorem}
  \label{thm:lti_stable}
  Let $\G$ be stable (i.e., $A$ is Hurwitz) and $\pi\in\R^{n_s}\rightarrow\R^{n_a}$ be a bounded causal controller. Assume that:
\begin{enumerate}
\item[(i)] the interconnection of $\G$ and $\pi$ is well-posed; 
\item[(ii)] $\pi$ has bounded partial derivatives on $\B$ (i.e., $\underline{\xi}_{ij}\leq\partial_j\pi_i(x)\leq\overline{\xi}_{ij}$, for all $x\in\B$, $i\in[n_a]$ and $j\in[n_s]$).
\end{enumerate}  
  If there exist $P\succeq 0$ and a scalar $\gamma>0$ such that $\text{SDP}(P,\lambda,\gamma,\xi)$ is feasible, then the feedback interconnection of $\G$ and $\pi$ is stable (i.e., it satisfies \eqref{equ:lti_stable}).
 \end{theorem}
\begin{proof}
The proof follows a standard dissipation argument. To proceed, we multiply $\begin{bmatrix}
x_{G}^\top&q^\top&e^\top
\end{bmatrix}^\top$ to the left and its transpose to the right of the augmented matrix in \eqref{equ:sdp_lti}, and use the constraints $w=Wq$ and $y=x_G$. Then, $\text{SDP}(P,\lambda,\gamma,\xi)$ can be written as a dissipation inequality:
\begin{equation*}
\dot{V}(x_{G})+\begin{bmatrix}
x_{G}\\
q
\end{bmatrix}^\top M(\lambda;\xi)\begin{bmatrix}
x_{G}\\
q
\end{bmatrix}< \gamma e^\top e-\frac{1}{\gamma}y^\top y,
\end{equation*}
where $V(x_G)=x_G^\top P x_G$ is known as the storage function, and $\dot{V}(\cdot)$ is its derivative with respect to time $t$. Because the second term is guaranteed to be non-negative by \cref{lem:quad_const}, if $\text{SDP}(P,\lambda,\gamma,\xi)$ is feasible with a solution $(P,\lambda,\gamma,\xi)$, we have:
\begin{equation}
\dot{V}(x_{G})+\frac{1}{\gamma}y^\top y-\gamma  e^\top e< 0,
\end{equation}
which is satisfied at all times $t$. From well-posedness, the above inequality can be integrated from $t=0$ to $t=T$, and then it follows from $P\succeq 0$ that:
\begin{equation}
\int_0^T|y(t)|^2dt\leq\gamma^2\int_0^T|e(t)|^2dt.
\end{equation}
Hence, the interconnected system with the RL policy $\pi$ is stable.
\end{proof}
The above theorem requires that $\G$ be stable when there is no feedback policy $\pi$. This is automatically satisfied in many physical systems with an existing stabilizing (but not performance-optimizing) controller. In the case that the original system is not stable, one needs to first design a controller to stablize the system or design the controller under uncertainty (in this case, the RL policy), which are well-studied problems in the literature (e.g., $H_\infty$ controller synthesis \cite{dullerud2013course}). Then, the result can be used to ensure stability while delegating reinforcement learning to optimize the performance of the policy under gradient bounds. 

The above result essentially suggests a computational approach in robust control analysis. Given a stable LTI system depicted in \eqref{equ:lti}, the first step is to represent the RL policy as an uncertainty block in a feedback interconnection. Because the parameters of the neural network policy may not be known {\it a priori} and will be continuously updated during learning, we characterize it using bounds on {partial gradients} (e.g., if it is known that the action is positively correlated with certain observation metric, we can specify its partial gradient to be mostly positive with only a small negative margin). A simple but conservative choice is a $\Log_2$-gain bound IQC; nevertheless, to achieve a less conservative result, we can employ the quadratic constraint developed in \cref{lem:quad_const}, which exploits both the sparsity of the control architecture and the non-homogeneity of the outputs. For a given set of gradient bounds $\xi$, we  find the smallest $\gamma$ such that \eqref{equ:sdp_lti} is feasible, and $\gamma$ corresponds to the upper bound on the $\Log_2$ gain of the interconnected system both during learning (with the excitation $e$ added to facilitate policy exploration) and actual deployment. If $\gamma$ is finite, then the system is provably stable in the sense of \eqref{equ:lti_stable}.

We remark that $\text{SDP}(P,\lambda,\gamma,\xi)$ is quasiconvex, in the sense that it reduces to a standard LMI with a fixed $\gamma$. To solve it numerically, we start with a small $\gamma$ and gradually increase it until a solution $(P,\lambda)$ is found. This is repeated for multiple sets of $\xi$. Each iteration (i.e., LMI for a given set of $\gamma$ and $\xi$) can be solved efficiently by interior-point methods. As an alternative to searching on $\gamma$ for a given $\xi$, more sophisticated methods for solving the generalized eigenvalue optimization problem can be employed \cite{boyd1994linear}.

\subsection{Extension to nonlinear systems with uncertainty}
\label{sec:nonlinear_analysis}

The previous analysis for LTI systems can be extended to a generic nonlinear system described in \eqref{equ:system}. The key idea is to model the nonlinear and potentially time-varying part $g_t(x(t))$ as an uncertain block with IQC constraints on its behavior. Specifically, consider the LTI component $\underline{\G}$:
\begin{equation}
\begin{cases}
\dot{x}_G&\!\!\!\!=Ax_G+Bu+v\\
y&\!\!\!\!=x_G
\end{cases}
\label{equ:nltv}
\end{equation}
where $x_G\in\R^{n_s}$ is the state and $y\in\R^{n_s}$ is the output. The linearized system is assumed to be stable, i.e., $A$ is Hurwitz. The nonlinear part is connected in {feedback}:
\begin{equation}
\begin{cases}
u&\!\!\!\!=e+w\\
w&\!\!\!\!=\pi(y)\\
v&\!\!\!\!=g_t(y)
\end{cases}
\label{equ:nltv_rl}
\end{equation}
where $e\in\R^{n_a}$ and $w\in\R^{n_a}$ are defined as before, and $g_t:\R^{n_s}\rightarrow\R^{n_s}$ is the nonlinear and time-varying component. In addition to characterizing $\pi$  using the Lipschitz property as in \eqref{equ:lip_con}, we assume that $g_t:\R^{n_s}\rightarrow\R^{n_s}$ satisfies the IQC defined by $(\bPsi,M_g)$ as in Definition \ref{def:iqc}. The system $\bPsi$ has the state-space representation:
\begin{equation}
\begin{cases}
\dot{\psi}&\!\!\!\!=A_\psi\psi+B_{\psi}^v v+B_{\psi}^y y\\
z&\!\!\!\!=C_\psi\psi+D_{\psi}^vv+D_{\psi}^yy
\end{cases},
\label{equ:nltv_psi}
\end{equation}
where $\psi\in\R^{n_s}$ is the internal state and $z\in\R^{n_z}$ is the {filtered} output. By denoting $x=\begin{bmatrix}
x_G^\top&\psi^\top
\end{bmatrix}^\top\in\R^{2n_s}$ as the new state, one can combine \eqref{equ:nltv} and \eqref{equ:nltv_psi} via reducing $y$ and letting $w=Wq$:
\begin{equation}
\begin{cases}
\dot{x}&\!\!\!\!=\underbrace{\begin{bmatrix}
A&0\\
B_\psi^y&A_\psi
\end{bmatrix}}_{\underline{A}}x+\underbrace{\begin{bmatrix}
B\\
0
\end{bmatrix}}_{\underline{B}_e}e+\underbrace{\begin{bmatrix}
BW\\
0
\end{bmatrix}}_{\underline{B}_q}q+\underbrace{\begin{bmatrix}
I\\
B^v_\psi
\end{bmatrix}}_{\underline{B}_v}v\\
z&\!\!\!\!=\underbrace{\begin{bmatrix}
D_\psi^y&C_\psi
\end{bmatrix}}_{\underline{C}}x+D_\psi^vv
\end{cases},
\label{equ:nltv_comb}
\end{equation}
where $\underline{A}$, $\underline{B}_e$, $\underline{B}_q$, $\underline{B}_v$, $\underline{C}$ are matrices of proper dimensions defined above. Similar to the case of LTI systems, the objective is to find the gradient bounds on $\pi$ such that the system becomes stable in the sense of \eqref{equ:lti_stable}. In the same vein, we define $\underline{\text{SDP}}(P,\lambda,\gamma,\xi)$ as:
\begin{equation}
\underline{\text{SDP}}(P,\lambda,\gamma,\xi):\begin{bmatrix}
O(P,\lambda,\xi) &O_v(P)& S(P)\\
O_v(P)^\top&D_\psi^{v\top}M_qD_\psi^v&0\\
      S(P)^\top &0& -\gamma I
\end{bmatrix}\prec 0,
\label{equ:nltv_sdp}
\end{equation}
where $P\succeq 0$, and
\begin{align*}
O(P,\lambda,\xi) &=\begin{bmatrix}
\ubA^\top P+P\ubA&P\ubB_q\\
\ubB_q^\top P&0
\end{bmatrix}+\begin{bmatrix}
\ubC^\top M_g\ubC&0\\
0&0
\end{bmatrix}+M(\lambda;\xi)+\frac{1}{\gamma}\begin{bmatrix}
I&0\\
0&0
\end{bmatrix},\\
O_v(P) &=\begin{bmatrix}
\ubC^\top M_q D_\psi^v+P\ubB_v\\
0
\end{bmatrix},\;
S(P)=\begin{bmatrix}
P\ubB_e\\
0
\end{bmatrix},
\end{align*}
where $M(\lambda;\xi)$ is defined in \eqref{equ:define_M}. The next theorem provides a stability certificate for the nonlinear time-varying system \eqref{equ:system}.
 
\begin{theorem}
  \label{thm:nltv_stable}
  Let $\underline{\G}$ be stable (i.e., $A$ in \eqref{equ:nltv} is Hurwitz) and $\pi\in\R^{n_s}\rightarrow\R^{n_a}$ be a bounded causal controller. Assume that:
\begin{enumerate}
\item[(i)] the interconnection of $\underline{\G}$, $\pi$, and $g_t$ is well-posed; 
\item[(ii)] $\pi$ has bounded partial derivatives on $\B$ (i.e., $\underline{\xi}_{ij}\leq\partial_j\pi_i(x)\leq\overline{\xi}_{ij}$ for all $x\in\B$, $i\in[n_a]$ and $j\in[n_s]$); 
\item[(iii)] $g_t\in\text{IQC}(\bPsi,M_g)$, where $\bPsi$ is stable.
\end{enumerate}  
  If there exist $P\succeq 0$ and a scalar $\gamma>0$ such that $\underline{\text{SDP}}(P,\lambda,\gamma,\xi)$ in \eqref{equ:nltv_sdp} is feasible, then the feedback interconnection of the nonlinear system \eqref{equ:system} and $\pi$ is stable (i.e., it satisfies \eqref{equ:lti_stable}).
 \end{theorem}
\begin{proof}
The proof is in the same vein as that of \cref{thm:lti_stable}. The main technical difference is the consideration of the filtered state $\psi$ and the output $z$ to impose IQC constraints on the nonlinearities $g_t(y)$ in the dynamical system \cite{megretski1997system}. The dissipation inequality follows by multiplying both sides of the matrix in \eqref{equ:nltv_sdp} by $\begin{bmatrix}
x^\top&q^\top&v^\top&e^\top
\end{bmatrix}^\top$ and its transpose:
\begin{equation*}
\dot{V}(x)+z^\top M_g z+\begin{bmatrix}
x_G\\
q
\end{bmatrix}^\top M_\pi \begin{bmatrix}
x_G\\
q
\end{bmatrix}<\gamma e^\top e-\frac{1}{\gamma} y^\top y,
\end{equation*}
where $x$ and $z$ are defined in \eqref{equ:nltv_comb}, and $V(x)=x^\top P x$ is the storage function with $\dot{V}(\cdot)$ as its time derivative. The second term on the left side is non-negative because $g_t\in\text{IQC}(\bPsi,M_g)$, and the third term is non-negative due to the smoothness quadratic constraint in \cref{lem:quad_const}. Thus, if there exists a feasible solution $P\succeq 0$ to $\underline{\text{SDP}}(P,\lambda,\gamma,\xi)$, integrating the inequality from $t=0$ to $t=T$ yields that:
\begin{equation}
\int_0^T|y(t)|^2dt\leq\gamma^2\int_0^T|e(t)|^2dt.
\end{equation}
Hence, the nonlinear system interconnected with the RL policy $\pi$ is certifiably stable in the sense of a finite $\Log_2$ gain.
\end{proof}

\subsection{Analysis of conservatism of the stability certificate}
\label{sec:necessity}

We focus on the case where an LTI system $\G$ is interconnected with an RL policy $\pi\in\mathcal{P}(\xi)$ (i.e., a function with bounded partial gradients). This corresponds to the system \eqref{equ:lti} studied in \cref{sec:compute}. To certify the stability of \eqref{equ:lti}, as will be shown in the next proposition, it suffices to examine the stability of the following system:
\begin{equation}
    \begin{cases}
    \dot{x}_G&\!\!\!\!=Ax_G+Bu\\
    q&\!\!\!\!=\tilde{\pi}(x_G)\\
    w&\!\!\!\!=Wq\\
    u&\!\!\!\!=e+w
    \end{cases}
    \label{equ:lti_rl_reform}
\end{equation}
where $\tilde{\pi}\in\widetilde{\mathcal{P}}$ is a function in the uncertainty set:
\begin{equation}
    \widetilde{\mathcal{P}}(\xi)=\left\lbrace \tilde{\pi}\;\middle| \;\underline{\xi}_{ij}x_j\leq{\tilde{\pi}_{ij}(x)}\leq\overline{\xi}_{ij}x_j,\forall x\in\R^{n_s}, i\in [n_a],j\in [n_s] \right\rbrace.
\end{equation}

\begin{proposition}
If the system \eqref{equ:lti_rl_reform} is stable for all $\tilde{\pi}\in\widetilde{\mathcal{P}}(\xi)$, then the system \eqref{equ:lti} is stable for all $\pi\in{\mathcal{P}}(\xi)$.
\label{prop:suff_cond_relation}
\end{proposition}
\begin{proof}
It suffices to show that for any $\pi\in{\mathcal{P}}(\xi)$, there exists a policy $\tilde{\pi}\in\widetilde{\mathcal{P}}(\xi)$ such that $\pi=W\tilde{\pi}$. Let $y^0_j=\begin{bmatrix}
0&\cdots&0&y_{j+1}&\cdots&y_{n_s}
\end{bmatrix}\in\R^{n_s}$ for every $j\in\{0,1,...,n_s\}$, and $y_0^0=y$, $y^0_{n_s}=0$. Then, one can write:
\begin{align*}
    \pi_i(y)=\sum_{j=1}^{n_s}\pi_i(y^0_{j-1})-\pi_i(y^0_j)=\sum_{j=1}^{n_s}\tilde{\pi}_{ij}(y),
\end{align*}
where $\tilde{\pi}_{ij}(y)$ satisfies
\begin{align*}
    \frac{\tilde{\pi}_{ij}(y)}{y_j}=\frac{\pi_i(y^0_{j-1})-\pi_i(y^0_j)}{|y^0_{j-1}-y^0_j|}\in[\underline{\xi}_{ij},\overline{\xi}_{ij}]
\end{align*}
if $y_j\neq 0$ and $\tilde{\pi}_{ij}(y)=0$ if $y_j=0$. The bound is due to the mean-value theorem and the bounds on the partial derivatives of $\pi_i$. Since the above argument is valid for all $i\in [n_a]$, it means that $\tilde{\pi}\in\widetilde{\mathcal{P}}(\xi)$, and $\pi=W\tilde{\pi}$.
\end{proof}

\cref{prop:suff_cond_relation} implies that one potential source of conservatism comes from the decomposition of a gradient-bounded function into a sum of sector-bounded components. Henceforth, we focus the subsequent analysis by examining \eqref{equ:lti_rl_reform}.  By considering the state-space representation of $G=
 \left[\begin{array}{c|cc}
  A& BW&B\\\hline
  I&0&0
   \end{array} \right]=
 \left[\begin{array}{cc}
  G_{11}& G_{12}
   \end{array} \right]$, one can write system \eqref{equ:lti_rl_reform} as:
   \begin{equation}
    \begin{cases}
    x_{G}&\!\!\!\!=\left[\begin{array}{cc}
  G_{11}& G_{12}
   \end{array} \right]\begin{bmatrix}q\\
   e\end{bmatrix}\\
    q&\!\!\!\!=\tilde{\pi}(x_G)\\
    \end{cases}.
    \label{equ:lti_rl_reform_freq}
    \end{equation}
    It is known that the system is input-output stable if and only if $I-G_{11}\tpi$ is nonsingular~\cite{dullerud2013course}. To understand this, note that if $I-G_{11}\tpi$ is nonsingular, then the transfer from $e$ to $x_{G}$ is given by: $$e\mapsto x_G=H(e)=(I-G_{11}\tpi)^{-1}G_{12}e,$$
    and $|x_{G}|=|H(e)|\leq\|(I-G_{11}\tpi)^{-1}G_{12}\||e|$. From the previous section (in particular, \cref{lem:quad_const}), we know that if the function $\pi$ is gradient-bounded, then the set of input/output signals belongs to:
    \begin{equation*}
        \mathcal{S}(\xi)=\left\lbrace(x,q)\mid\phi_{ij}(x,q)=(\overline{c}_{ij}^2-c_{ij}^2) x_j^2+2c_{ij}q_{ij}x_j-q_{ij}^2\geq 0,\quad\forall i\in[n_a],j\in[n_s]\right\rbrace,
    \end{equation*}
where we use $c_{ij}= \frac{1}{2}\left(\underline{\xi}_{ij}+\overline{\xi}_{ij}\right)$, $\overline{c}_{ij}=\overline{\xi}_{ij}-c_{ij}$ for simplicity. We now show that the pair $(x,q)$ belongs to $\mathcal{S}(\xi)$ if and only if there exists a sector-bounded function $\tpi\in\tilde{P}(\xi)$ such that it satisfies $q=\tpi(x)$.

\begin{lemma}
Suppose that $x\in\R^{n_s}$ and $q\in\R^{n_an_s}$, and $\overline{c}_{ij}\geq 0$ for every $i\in[n_a]$ and $j\in[n_s]$. Then, the pair $(x,q)$ belongs to $\mathcal{S}(\xi)$ if and only if there exists an operator $\tpi:\R^n\rightarrow\R^{n_an_s}$, such that $q=\tpi(x)$, and $\tpi$ satisfies the following conditions: \textbf{(i)} $\tpi_{ij}(x)=0$ if $x_j=0$, and \textbf{(ii)} $\tpi$ is sector bounded, i.e.,  $(c_{ij}-\overline{c}_{ij})x_j\leq{\tpi_{ij}(x)}\leq(\overline{c}_{ij}+c_{ij})x_j$ holds for all $i\in[n_a]$ and $j\in[n_s]$.
\end{lemma}
\begin{proof} To show the sufficiency direction, conditions (i) and (ii) yield that
\begin{align*}
    \left(\overline{c}_{ij}x_j\right)^2
    &\geq \left(\left(\frac{\pi_{ij}({x})}{x_j}-c_{ij}\right){x}_j\right)^2=\left(q_{ij}-c_{ij}x_j\right)^2.
\end{align*}
By rearranging the above inequality, it can be concluded that $(x,q)\in\mathcal{S}(\xi)$. 

For the necessary direction, note that the condition $\phi_{ij}(x,q)\geq 0$ is equivalent to $\left|q_{ij}-c_{ij}x_j\right|\leq \left|\overline{c}_{ij}x_j\right|$. Thus, we have $\pi_{ij}(x)=0$ if $x_j=0$. Since $\overline{c}_{ij}\geq 0$, one can obtain $\left|\frac{q_{ij}}{x_j}-c_{ij}\right|\leq \overline{c}_{ij}$, which is equivalent to the sector bounds.
\end{proof}

By slightly overloading the notations, we can extend the result of the previous lemma from static mapping to the case that $x\in L^{n_s}$ and $q\in L^{n_an_s}$ with the operator $\tpi:L^{n_s}\rightarrow L^{n_an_s}$. We can then extend the definition of $\mathcal{S}(\xi)$ to this space accordingly. 

\begin{lemma}
\label{lem:equiv_delta}
Suppose that $x\in L^{n_s}$ and $q\in L^{n_an_s}$, and that $\overline{c}_{ij}\geq 0$ for all $i\in[n_a]$ and $j\in[n_s]$. Then, the pair $(x,q)$ belongs to $\mathcal{S}(\xi)$ where
\begin{equation}
    \mathcal{S}(\xi)=\left\lbrace(x,q)\mid\phi_{ij}(x,q)\geq 0,\quad\forall i\in[n_a],j\in[n_s]\right\rbrace,
\end{equation}
and \begin{equation}
    \phi_{ij}(x,q)=(\overline{c}_{ij}^2-c_{ij}^2) \|x_j\|^2+2c_{ij}\langle q_{ij},x_j\rangle-\|q_{ij}\|^2\geq 0,
    \label{equ:phi_def}
\end{equation}
if and only if there exists an operator $\tpi:L^{n_s}\rightarrow L^{n_an_s}$ such that $q=\tpi(x)$ and $\tpi$ satisfies the following conditions: \textbf{(i)} $\tpi_{ij}(x)=0$ if $x_j=0$, and \textbf{(ii)} $\tpi$ is sector bounded, i.e.,  $(c_{ij}-\overline{c}_{ij}){\|x_j\|}\leq{\left \|{\tpi_{ij}({x})}\right\|}\leq(\overline{c}_{ij}+c_{ij}){\|x_j\|}$ for all $i\in[m]$ and $j\in[n]$.
\end{lemma}
\begin{proof}For the sufficiency condition, since $\tpi_{ij}$ is sector bounded, and $\tpi_{ij}(x)=0$ if $x_j=0$, without loss of generality, assume that $c_{ij}\leq 0$. One can write
\begin{align*}
    \left\|\overline{c}_{ij}x_j\right\|^2
    &\geq \left\|\left|\frac{\left\|{\tpi_{ij}({x})}\right\|}{\|x_j\|}-c_{ij}\right|x_j \right\|^2\\
    &\geq \left\|\left(\frac{\left\|{\tpi_{ij}({x})}\right\|}{\|x_j\|}-c_{ij}\right)x_j \right\|^2\\
    &=\|c_{ij}x_j\|^2+\|\tpi_{ij}(x)\|^2-2\left\langle c_{ij}x_j,\frac{\left\|{\tpi_{ij}({x})}\right\|}{\|x_j\|}x_j\right\rangle\\
    &=\|c_{ij}x_j\|^2+\|\tpi_{ij}(x)\|^2-2c_{ij}
    \|x_j\|\left\|{\tpi_{ij}({x})}\right\|\\
    &\geq \|c_{ij}x_j\|^2+\|\tpi_{ij}(x)\|^2-2c_{ij}
    \left\langle{\tpi_{ij}({x}),x_j}\right\rangle\\
    &=\left\|q_{ij}-c_{ij}{x}_j\right\|^2.
\end{align*}
By rearranging the above inequality, it can be concluded that $(x,q)\in\mathcal{S}(\xi)$.

For the necessary direction, we can construct $\tpi(y) = q\frac{\langle y,x\rangle}{|x|^2}$ for all $y\in L^{n_s}$. This leads to $\tpi(x)=q$, and the condition $\phi_{ij}(x,q)\geq 0$ is equivalent to $\left\|q_{ij}-c_{ij}x_j\right\|\leq \overline{c}_{ij}\left\|x_j\right\|$. Thus, we have $\tpi_{ij}(x)=0$ if $x_j=0$. Without loss of generality, assume that $c_{ij}\leq 0$. Therefore, $\|q_{ij}\|\leq \overline{c}_{ij}\|x_j\|+c_{ij}\|x_j\|$ and $\|q_{ij}\|\geq -\overline{c}_{ij}\|x_j\|+c_{ij}\|x_j\|$, which are equivalent to the sector bound condition.
\end{proof}

The previous result indicates that the input and output pair of $\tpi$ can be described by $\mathcal{S}(\xi)$. We show next that this set should be separated from the signal space of the dynamical system in order to ensure robust stability. 
\begin{lemma}
If $(G,\tpi)$ is robustly stable, then there cannot exist a nonzero $q\in L_2$ such that $x=Gq$ and $(x,q)\in\mathcal{S}(\xi)$.
\label{lem:stable_set}
\end{lemma}
\begin{proof}
We prove this lemma by contraposition. If there exists a nonzero $q\in L_2$ such that $(x,q)\in\mathcal{S}(\xi)$, then it follows from \cref{lem:equiv_delta} that there exists a linear operator $\tpi$ such that $q=\tpi(x)=\tpi(Gq)$. This implies that the operator $(I-\tpi G)$ is singular, and therefore, $(I-G\tpi)$ is singular, implying that the interconnected system is not robustly stable. 
\end{proof}

The path of examining the necessity of the SDP condition \eqref{equ:sdp_lti} has become clear. Consider the set generated by the LTI system:
\begin{equation}
    \Psi=\left\lbrace(\phi_{ij}(x,q):q\in L^{n_an_s},\|q\|=1,x=Gq\right\rbrace,
    \label{equ:psi_set}
\end{equation}
and the positive orthant
\begin{equation}
    \Pi=\left\lbrace (r_{ij})\in\R^{n_an_s}:r_{ij}\geq 0,\quad \forall i\in[n_a], j\in[n_s]\right\rbrace.
\end{equation}
 \cref{lem:stable_set} implies that the two sets $\Psi$ and $\Pi$ are separated if $(G,\tpi)$ is robustly stable. The goal is to show that there exists a separating hyperplane, whose parameters are related to the solution of \eqref{equ:sdp_lti}. For simplicity, define the matrices ${\Omega}_{ij,x}=\text{diag}\left(\left\lbrace\lbrace \overline{c}_{ij}^2-c_{ij}^2\right\rbrace\right)$, ${\Omega}_{ij,q}$ and ${\Omega}_{ij,xq}$ with their $(k,l)$-th elements $[{\Omega}_{ij,q}]_{kl}=\begin{cases}
1&\text{ if }k=in+j\\
0&\text{ otherwise }
\end{cases}$, and $[{\Omega}_{ij,xq}]_{kl}=\begin{cases}
c_{ij}&\text{ if }k=j, l=in+j\\
0&\text{ otherwise }
\end{cases}$. To write $\phi_{ij}(x=Gq,q)$ as an inner product, define $$T_{ij}=G^*{\Omega}_{ij,x}G-{\Omega}_{ij,q}+G^*{\Omega}_{ij,xq}^*+{\Omega}_{ij,xq}G.$$
It results from the definition \eqref{equ:phi_def} that \begin{equation}
    \phi_{ij}(x=Gq,q)=\|Gq\|^2_{{\Omega}_{ij,x}}+2\text{Re}\left\langle {\Omega}_{ij,xq}Gq,q\right\rangle-\|q\|^2_{{\Omega}_{ij,q}}=\left\langle q,T_{ij}q\right\rangle.
\end{equation}

\begin{lemma}
For a given linear time-invariant operator $G$, the closure $\overline{\Psi}$ of $\Psi$ defined in \eqref{equ:psi_set} is convex.
\label{lem:convex_psi}
\end{lemma}
\begin{proof}
 Because $G$ is time-invariant, by denoting $D_\tau$ as the delay operator at scale $\tau$, we obtain $D_\tau^* T_{ij}D_\tau=T_{ij}$. Let $y=\phi(q)$ and $\tilde{y}=\phi(\tilde{q})$ be the elements of $\Psi$, with $\|q\|=\|\tilde{q}\|=1$. By considering $q_\tau=\sqrt{\alpha}q+\sqrt{1-\alpha}D_\tau\tilde{q}$, one can write
\begin{align*}
\phi_{ij}(q_\tau)&=\alpha\left\langle T_{ij}q,q\right\rangle+(1-\alpha)\left\langle T_{ij}D_\tau\tilde{q},D_\tau\tilde{q}\right\rangle+2\alpha\sqrt{1-\alpha}\text{Re}\left\langle T_{ij}q,D_\tau\tilde{q}\right\rangle\\
&=\alpha\phi_{ij}(q)+(1-\alpha)\phi_{ij}(\tilde{q})+2\alpha\sqrt{1-\alpha}\text{Re}\left\langle T_{ij}q,D_\tau\tilde{q}\right\rangle.
\end{align*}
By letting $\tau\rightarrow\infty$, we obtain $\text{Re}\left\langle T_{ij}q,D_\tau\tilde{q}\right\rangle\rightarrow 0$, where $\text{Re}(x)$ denotes the real part of a complex vector $x$. Thus,
$$\lim_{\tau\rightarrow\infty}\phi_{ij}(q_\tau)=\alpha\phi_{ij}(q)+(1-\alpha)\phi_{ij}(\tilde{q})$$
and $\lim_{\tau\rightarrow\infty}\|q_\tau\|^2=\alpha\|q\|^2+(1-\alpha)\|\tilde{q}\|^2=1$. Therefore, 
$$\lim_{\tau\rightarrow\infty}\phi\left(\frac{q_\tau}{\|q_\tau\|}\right)=\alpha y+(1-\alpha)\tilde{y}\in\overline{\Psi}.$$
\end{proof}

Now, we show that strict separation occurs when the system is robustly stable.
\begin{lemma}
Suppose that $I-G\tpi$ is nonsingular. Then, the sets $\Pi$ and $\Psi$ are strictly separated, namely $D(\Pi,\Psi)=\inf_{{r}\in\Pi,{y}\in\Psi}\left|{r}-{y}\right|>0$.
\label{lem:strict_separate}
\end{lemma}

To prove this result, we need the following lemma.

\begin{lemma}
\label{lem:d0case}
Suppose that $D(\Pi,\Psi)=\inf_{{r}\in\Pi,{y}\in\Psi}\left|{r}-{y}\right|=0$. Given any $\epsilon>0$ and $t_0\geq 0$, there exist a closed  interval $[t_0,t_1]$ and two signals $x\in L^{n_s}$ and $q\in L^{n_an_s}$ with $\|q\|=1$ such that 
\begin{align}
    \phi_{ij}(x,q)&\geq 0, \qquad\forall i\in[n_a],j\in[n_s]\label{equ:phi_ineq}\\ \epsilon^2&>\|(I-\Gamma_{[t_0,t_1]})Gq\|\label{equ:trunc_eps}\\
    \epsilon&=\|x-\Gamma_{[t_0,t_1]}Gq\|_{{\Omega}_{ij,x}}\label{equ:trunc_x},
\end{align}
where $\|q\|_{\Omega}=\sqrt{q^*{\Omega}q}$ is the scaled norm and $\Gamma_{[t_0,t_1]}$ projects the signal onto the support of $[t_0,t_1]$. With the above choice of $q,x$ and $[t_0,t_1]$, there exists an operator $\tpi\in\tilde{P}(\xi)$ such that $\|(I-\tpi \Gamma_{[t_0,t_1]}G)q\|\leq C\epsilon$ for some constant $C>0$ that depends on the sector bounds $\xi$.
\end{lemma}
\begin{proof}
For a given $\epsilon>0$, by hypothesis, there exists $q\in L^{n_an_s}$ with $\|q\|=1$ satisfying $\phi_{ij}(x,q)>-\epsilon^2$ for all $i\in[n_a]$ and $j\in[n_s]$, i.e.,
$$\epsilon^2+\|Gq\|^2_{{\Omega}_{ij,x}}+2\text{Re}\left\langle {\Omega}_{ij,xq}Gq,q\right\rangle>\|q\|^2_{{\Omega}_{ij,q}},$$
where ${\Omega}_{ij,x}$ and ${\Omega}_{ij,xq}$ are defined previously. Clearly, if $q$ is truncated to a sufficiently long interval, and $q$ is rescaled to have a unit norm, the above inequality will still hold. Since $Gq\in L^{n_s}$, by possibly enlarging the truncation interval to $[t_0,t_1]$, we obtain \eqref{equ:trunc_eps}, and
$$\epsilon^2+\|\Gamma_{[t_0,t_1]}Gq\|^2_{{\Omega}_{ij,x}}+2\text{Re}\left\langle {\Omega}_{ij,xq}\Gamma_{[t_0,t_1]}Gq,q\right\rangle>\|q\|^2_{{\Omega}_{ij,q}},$$
Next, we choose $\eta\in L^{n_s}$ such that $\|\eta\|^2_{{\Omega}_{ij,x}}=\epsilon^2$, and that $\eta$ is orthogonal to $\Gamma_{[t_0,t_1]}Gq$ and ${\Omega}_{ij,xq}^*q$ for all $i\in[n_a]$ and $j\in[n_s]$. Then, by considering $x=\Gamma_{[t_0,t_1]}Gq+\eta$, we obtain
$$\|x\|^2_{{\Omega}_{ij,x}}=\|\Gamma_{[t_0,t_1]}Gq+\eta\|^2_{{\Omega}_{ij,x}}=\epsilon^2+\|\Gamma_{[t_0,t_1]}Gq\|^2_{{\Omega}_{ij,x}},$$
which leads to $\phi_{ij}(x,q)\geq 0$ and \eqref{equ:trunc_x}. Now, we can invoke Lemma \ref{lem:equiv_delta} to construct $\tpi\in\mathcal{P}(\xi)$ based on \eqref{equ:phi_ineq} such that $\tpi$ becomes sector bounded and $q=\tpi x$. Then,
$$(I-\tpi \Gamma_{[t_0,t_1]}G)q=\tpi(x-\Gamma_{[t_0,t_1]}Gq).$$
Let $\|\tpi\|\leq C$ (which depends on the sector bounds). Then,
$$\|(I-\tpi \Gamma_{[t_0,t_1]}G)q\|\leq C\epsilon$$.
\end{proof}

We are now ready to prove the strict separation result.
\begin{proof}[Proof of \cref{lem:strict_separate}]
Assume that $D(\Pi,\Psi)=\inf_{{r}\in\Pi,{y}\in\Psi}\left|{r}-{y}\right|=0$. Consider a sequence $\epsilon_n\rightarrow 0$ as $n$ tends to $\infty$. For each $\epsilon_n$, construct signals $q^{(n)}$ with a bounded support on $[t_n,t_{n+1}]$, and $\tpi^{(n)}$ according to Lemma \ref{lem:d0case}. Define $$\tpi=\sum_{n=1}^\infty \tpi^{(n)}\Gamma_{[t_n,t_{n+1}]}.$$
We have 
\begin{align*}
    \tpi Gq^{(n)}&=\tpi^{(n)}\Gamma_{[t_n,t_{n+1}]}Gq^{(n)}+\tpi(I-\Gamma_{[t_n,t_{n+1}]})Gq^{(n)},
\end{align*}
and 
\begin{align*}
    \|(I-\pi G)q^{(n)}\|&\leq \|(I-\tpi^{(n)}\Gamma_{[t_n,t_{n+1}]}G)q^{(n)}\|+\|(I-\Gamma_{[t_n,t_{n+1}]})Gq^{(n)}\|\\
    &\leq C\epsilon_n+\epsilon_n^2
\end{align*}
Because $\epsilon_n\rightarrow 0$, the right-hand side approaches 0, and so does the left-hand side. Therefore, since $\|q^{(n)}\|=1$, the mapping $I-\tpi G$ cannot be invertible, which contradicts the robust stability assumption. This implies that $\Pi$ and $\Psi$ are strictly separable.
\end{proof}

To draw the connection to the SDP problem \eqref{equ:sdp_lti}, observe that 
\begin{equation}
    \phi_{ij}(x,q)=\left\langle\begin{bmatrix}
    x\\
    q
    \end{bmatrix},M_\pi^{ij}\begin{bmatrix}
    x\\
    q
    \end{bmatrix}\right\rangle,
\end{equation}
where $$[M_\pi^{ij}]_{kl}=\begin{cases}
\overline{c}_{ij}^2-c_{ij}^2&(k,l)=(j,j)\\
c_{ij}&(k,l)=(i,i*n+j) \text{ or }(i*n+j,i)\\
-1&(k,l)=(i*n+j,i*n+j)
\end{cases},$$ and $M(\lambda;\xi)=\sum_{i\in[n_a],j\in[n_s]}\lambda_{ij} M_\pi^{ij}$ as defined in \cref{lem:quad_const}. 

\begin{proposition}
\label{prop:equi_sdp_s2}
The SDP condition \eqref{equ:sdp_lti} is feasible  if and only if there exist multipliers $\lambda_{ij}\geq 0$ and $\epsilon>0$ such that 
\begin{equation}
    \sum_{i\in[n_a],j\in[n_s]}\lambda_{ij}\phi_{ij}(x,q)\leq-\epsilon\|q\|^2
    \label{equ:s2_cond}
\end{equation}
for all $q\in L^{n_an_s}$ and $x=Gq$. 
\end{proposition}
\begin{proof}
Since $\phi_{ij}(x,q)=\left\langle\begin{bmatrix}
    x\\
    q
    \end{bmatrix},M_\pi^{ij}\begin{bmatrix}
    x\\
    q
    \end{bmatrix}\right\rangle,$ the condition \eqref{equ:s2_cond} is equivalent to 
    \begin{equation}
        \begin{bmatrix}
        G\\
        I
        \end{bmatrix}^*M(\lambda;\xi)\begin{bmatrix}
        G\\
        I
        \end{bmatrix}\prec 0.
    \end{equation}
    By the KYP lemma, this is equivalent to the existence of $P\succeq 0$ such that:
    \begin{equation}
        \begin{bmatrix}
        A^\top P+PA&PBW\\
        W^\top B^\top P&0
        \end{bmatrix}+M(\lambda;\xi)\prec 0.
    \end{equation}
    By Schur complement, $P$ satisfies the KYP condition if and only if it satisfies \eqref{equ:sdp_lti}. Thus, the claim is shown.
\end{proof}

\begin{theorem}
  \label{thm:ltv_necessity}
  Let  $\tpi:L^{n_s}\rightarrow L^{n_an_s}$ be a bounded causal controller such that $\tpi\in\tilde{\mathcal{P}}(\xi)$. Assume that the interconnection of $G$ and $\tpi$ is well-posed. Then, the input-output stability of the feedback interconnection of system \eqref{equ:lti_rl_reform} implies that there exist $P\succeq 0$, $\gamma>0$ and $\lambda\geq 0$ such that $\text{SDP}(P,\lambda,\gamma,\xi)$ in \eqref{equ:sdp_lti} is feasible.
 \end{theorem}
 \begin{proof}
 Since the system is input-output stable, the sets $\Pi$ and $\overline{\Psi}$ are strictly separable due to \cref{lem:strict_separate}. Since both $\Pi$ and $\overline{\Psi}$ are convex (\cref{lem:convex_psi}), there exist a strictly separating hyperplane parametrized by $\lambda\in\R^{mn}$ and scalars $\alpha,\beta$, such that $$\langle\lambda,\phi\rangle\leq\alpha<\beta\leq\langle\lambda,y\rangle$$ for all $\phi\in\overline{\Psi}$ and $y\in\Pi$. Since $\langle\lambda,y\rangle$ is bounded from below, we must have $\lambda\geq 0$, and without loss of generality, we can set $\beta=0$ and $\alpha<0$. This condition is equivalent to \eqref{equ:s2_cond}, and by Proposition \ref{prop:equi_sdp_s2}, this implies that the SDP condition is feasible.
 \end{proof}

\section{Numerical examples}
\label{sec:experiments}

In this section, we empirically study the stability-certified reinforcement learning in real-world problems such as flight formation \cite{hauser1992nonlinear} and power grid frequency regulation \cite{fazelnia2017convex}. Designing an optimal controller for these systems is challenging, because they consist of many interconnected subsystems that have limited information sharing, and also their underlying models are typically nonlinear and even time-varying and uncertain. Indeed, for the case of decentralized control, which aims at designing a set of local controllers whose interactions are specified by physical and informational structures, it has been long known that it amounts
 to an NP-hard optimization problem in general \cite{bakule2008decentralized}. End-to-end reinforcement learning comes in handy, because it does not require model information by simply interacting with the environment while collecting rewards. 

In a multi-agent setting, each agent explores and learns its own policy independently without knowing about other agents' policies \cite{bucsoniu2010multi}. For the simplicity of implementation, we consider the synchronous and cooperative scenario, where agents conduct an action at each time step and observe the reward for the whole system. Their goal is to collectively maximize the rewards (or minimize the costs) shared equally among them. The present analysis aims at \emph{offering safety certificates of existing RL algorithms when applied to real-world dynamical systems}, by simply monitoring the gradients information of the neural network policy. This is orthogonal to the line of research that aims at improving the performance of the existing RL algorithms. The examples are taken from \cite{hauser1992nonlinear,fazelnia2017convex,fattahi2017transformation}, but we deal directly with the underlying \emph{nonlinear} physics rather than a linearized model.

\subsection{Multi-agent flight formation}

Consider the multi-agent flight formation problem \cite{hauser1992nonlinear}, where each agent can only observe the relative distance from its neighbors, as illustrated in \cref{fig:flight_illustrate}. The goal is to design a local controller for each aircraft such a predefined pattern is formed as efficiently as possible.
\begin{figure}[h!]
  \centering
  \includegraphics[width=0.55\textwidth]{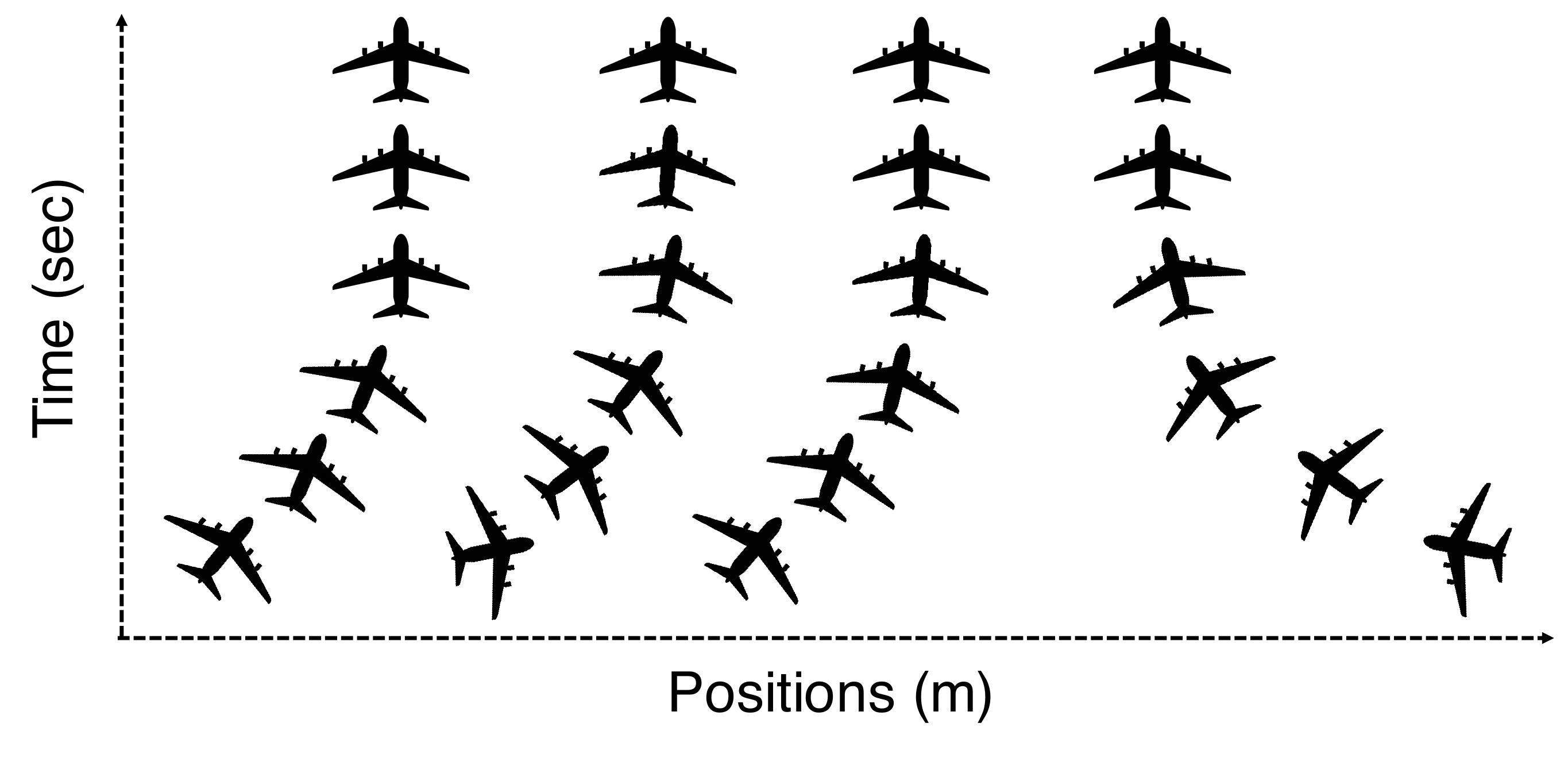}
  \caption{Illustration of the multi-agent flight formation problem.
  }
  \label{fig:flight_illustrate}
\end{figure}
The physical model\footnote{The cosine term in the original formulation is omitted for simplicity, though it can be incorporated in a more comprehensive treatment.} for each aircraft is given by:
\begin{align*}
\ddot{z}^i(t)&=v^i(t)\\
\ddot{\theta}^i(t)&=\frac{1}{\delta}\left(\sin\theta^i(t)+v^i(t)\right),
\end{align*}
where $z^i$ and $\theta^i$ denote the horizontal position and angle of aircraft $i$, respectively, and $\delta>0$ characterizes the physical coupling of rolling moment and lateral acceleration. To stabilize the system, a simple feedback rule is proposed in \cite{arcak2012synchronization},
\begin{equation}
v^i(t)=\alpha \dot{z}^i(t)+\beta\theta^i(t)+\gamma\dot{\theta}^i(t)+u^i(t)
\end{equation}
where the parameters of the first three terms are designed to maintain the internal stability of the horizontal speed and angle of each aircraft (specifically, $\alpha=90.62$, $\beta=-42.15$, $\gamma=-13.22$, $\delta=0.1$ as explained in \cite{arcak2012synchronization}), and the last term is an external input optimized for performance (e.g., to move the aircraft to a target state as fast as possible). For each agent, by defining the state $x^i(t)=\begin{bmatrix}
\dot{z}^i(t)&\theta^i(t)&\dot{\theta}(t)
\end{bmatrix}^\top$, the above dynamics can be written as
\begin{equation}
\dot{x}^i(t)=\underbrace{\begin{bmatrix}
\alpha&\beta&\gamma\\
0&0&1\\
\frac{\alpha}{\delta}&\frac{\beta+1}{\delta}&\frac{\gamma}{\delta}
\end{bmatrix}}_{A^i}x^i(t)+\underbrace{\begin{bmatrix}
1\\0\\\frac{1}{\delta}
\end{bmatrix}}_{B^i}u^i(t)+\underbrace{\begin{bmatrix}
0\\0\\\sin\theta^i(t)-\theta^i(t)
\end{bmatrix}}_{g^i(\bx^i(t))},
\end{equation}
where $g^i(x^i(t))$ is a nonlinear function of $x^i(t)$ that is neglected for a linearized model \cite{arcak2012synchronization,fattahi2017transformation}. In a distributed control setting, each agent only has access to the relative distance from its neighbors; therefore, for agents $i=1,2,3$, define
\begin{equation}
\widetilde{x}^i(t)=\begin{bmatrix}
z^i(t)-z^{i+1}(t)-d&x^i(t)^\top
\end{bmatrix}^\top,
\end{equation}
where $d$ is the desired distance between agents. The state-space model of the interconnected system can be written in the form of \eqref{equ:system}:
\begin{equation}
\begin{bmatrix}
\dot{\widetilde{x}}^1\\
\dot{\widetilde{x}}^2\\
\dot{\widetilde{x}}^3\\
\dot{x}^4
\end{bmatrix}=\underbrace{\begin{bmatrix}
\widetilde{A}^1&H_4&0&0\\
0&\widetilde{A}^2&H_4&0\\
0&0&\widetilde{A}^3&H_3\\
0&0&0&A^4
\end{bmatrix}}_{A}\underbrace{\begin{bmatrix}
\widetilde{x}^1\\
\widetilde{x}^2\\
\widetilde{x}^3\\
x^4
\end{bmatrix}}_{x(t)}+
\underbrace{\begin{bmatrix}
\widetilde{B}^1&0&0&0\\
0&\widetilde{B}^2&0&0\\
0&0&\widetilde{B}^3&0\\
0&0&0&B^4
\end{bmatrix}}_{B}\underbrace{\begin{bmatrix}
u^1\\
u^2\\
u^3\\
u^4
\end{bmatrix}}_{u(t)}+\underbrace{\begin{bmatrix}
\widetilde{g}^1(x^1)\\
\widetilde{g}^2(x^2)\\
\widetilde{g}^3(x^3)\\
{g}^4(x^4)
\end{bmatrix}}_{g(x(t))},
\label{equ:pvtol_system}
\end{equation}
where $H_3$ (or $H_4$) is a $4\times 3$ (or $4\times 4$) matrix whose $(i,j)^{\text{th}}$ entry is equal to $-1$ if $(i,j)=(1,1)$ (or $(i,j)=(1,2)$) and is zero otherwise, and where $\widetilde{A}^i$, $\widetilde{B}^i$ and $\widetilde{g}^i(x^i(t))$ for $i=1,2,3$ are augmented to account for the state of relative positions, given by
\begin{equation}
\widetilde{A}^i=\begin{bmatrix}
0&1&0&0\\
0&\alpha&\beta&\gamma\\
0&0&0&1\\
0&\frac{\alpha}{\delta}&\frac{\beta+1}{\delta}&\frac{\gamma}{\delta}
\end{bmatrix},\;\;\widetilde{B}^i=\begin{bmatrix}
0\\1\\0\\\frac{1}{\delta}
\end{bmatrix},\;\;\widetilde{g}^i(x^i(t))=\begin{bmatrix}
0\\0\\0\\\sin\theta^i(t)-\theta^i(t)
\end{bmatrix}.
\end{equation}
One particular strength of RL is that the reward function can be highly nonconvex, nonlinear, and arbitrarily designed; however, since quadratic costs are widely used in the control literature, consider the case $r(x(t),u(t))=x(t)^\top Qx(t)+u(t)^\top Ru(t)$. For the following experiments, assume that $Q=1000\times I_{15}$ and $R=I_4$. In addition, because the original system $A$ has its largest eigenvalue at 0, we need a nominal distributed linear controller $K_d$, whose primary goal is to make the largest eigenvalue of $A+BK_n$ negative. Such controller could be designed using methods such as robust control synthesis for the linearized system \cite{dullerud2013course,zhou1996robust}. 
With the{ nominal controller} in place, we can define the new system matrix $A_G=A+BK_n$ and replace $A$ in \eqref{equ:pvtol_system}. 

The task for multi-agent RL is to learn the controller $u^i(t)$, which only takes inputs of the relative distances of agent $i$ to its neighbors. For example, agent 1 can only observe $z^1(t)-z^2(t)-d$ (i.e., the 1$^\text{st}$ entry of ${x}(t)$); similarly, agent 2 can only observe $z^1(t)-z^2(t)-d$ and $z^2(t)-z^3(t)-d$ (i.e.,  the 1$^\text{st}$ and 5$^\text{th}$ entries of $x(t)$). 

\textbf{Stability certificate:} To obtain the stability certificate of \eqref{equ:pvtol_system}, we apply the method in \cref{sec:nonlinear_analysis}. The nonzero entries of the nonlinear component $g(x(t))$ are in the form of $\sin(\theta)-\theta$, which can be treated as an uncertainty block with the slope restricted to $[-1,0]$ for $\theta\in[-\frac{\pi}{2},\frac{\pi}{2}]$; therefore, the Zames-Falb IQCs can be employed to construct \eqref{equ:nltv_psi} \cite{zames1968stability,lessard2016analysis}. As for the RL agents $u^i$, their gradient bounds can be certified according to Theorem \ref{thm:nltv_stable}. Specifically, we assume that each agent $u^i$ is $l$-Lipschitz continuous, and solve \eqref{equ:nltv_sdp} for a given set of $\gamma$ and $l$. The certified gradient bounds (Lipschitz constants) are plotted in \cref{fig:lip_gains_pvtol} using different constraints. The conservative $L_2$ constraint \eqref{equ:ineq_const_1} is only able to certify stability for Lipschitz constants up to 0.8. By incorporating the sparsity of distributed controller, we can increase the margin to 1.2, which is satisfied throughout the learning process.

\begin{figure}[h!]
  \centering
  \includegraphics[width=0.7\textwidth]{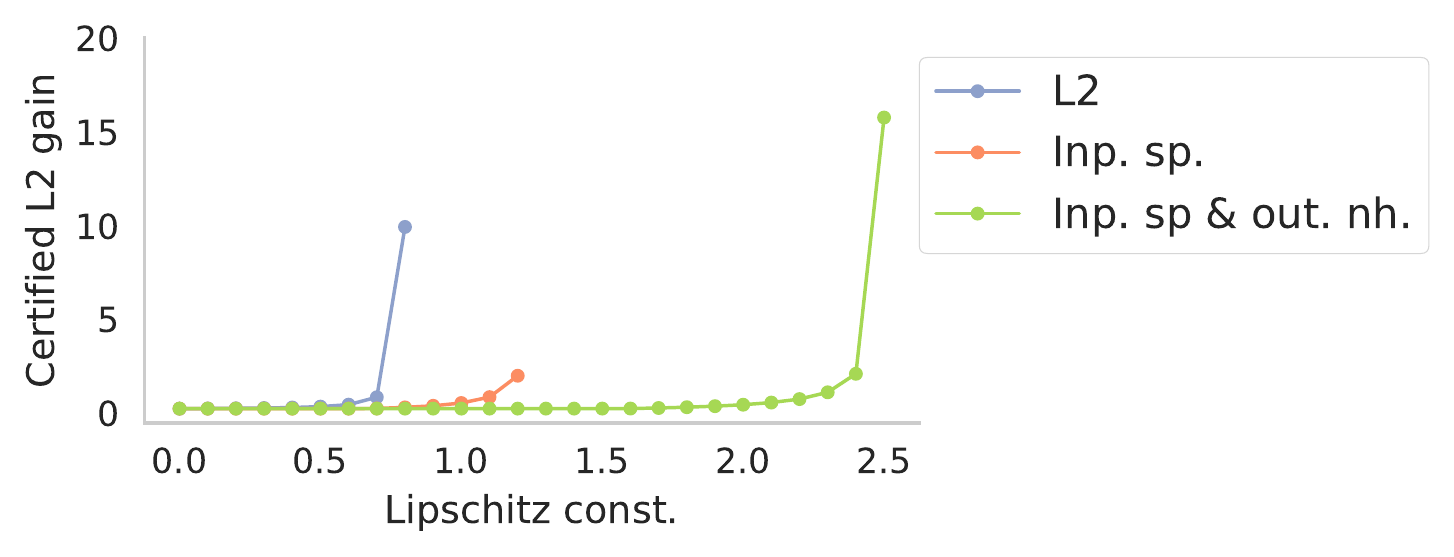}
  \caption{Stability-certified Lipschitz constants obtained by the standard $L_2$ bound (L2) in \eqref{equ:ineq_const_1} and the method proposed in \cref{lem:quad_const}, which considers input sparsity (inp. sp.) and output non-homogeneity (out. nh.).
  }
  \label{fig:lip_gains_pvtol}
\end{figure}

In order to further increase the set of certifiable stable controllers, we monitor the partial gradient information for each agent and encode them as non-homogeneous gradient bounds. For instance, if $\frac{\partial \pi_i(x)}{\partial x_j}$ has been consistently positive for latest iterations, we will set $\overline{\xi}_{ij}=l$ and $\underline{\xi}_{ij}=-\epsilon l$, where $\epsilon>0$ is a small margin, such as 0.1, to allow explorations. By performing this during learning, it would be possible to significantly enlarge the certified Lipschitz bound to up to 2.5, as shown in \cref{fig:lip_gains_pvtol}.

\textbf{Policy gradient RL:} To perform multi-agent reinforcement learning, we employ trust region policy optimization with natural gradients and smoothness policies. During learning, we employ the hard-thresholding step introduced in \cref{sec:rl_polgrad} to ensure that the gradient bounds are satisfied. The trajectories of rewards averaged over three independent experiments are shown in \cref{fig:rew_iters_pvtol}. In this example, agents with a 1-layer neural network (each with 5 hidden units) can learn most efficiently when employed with the smoothness penalties (coefficients are set to be $\omega_1=\omega_2=0.01$) in \eqref{equ:pol_obj}. Without the guidance of these penalties, the linear controller and 1-layer neural network apparently cannot effectively explore the parameter space. 

\begin{figure}[h!]
  \centering
  \includegraphics[width=0.7\textwidth]{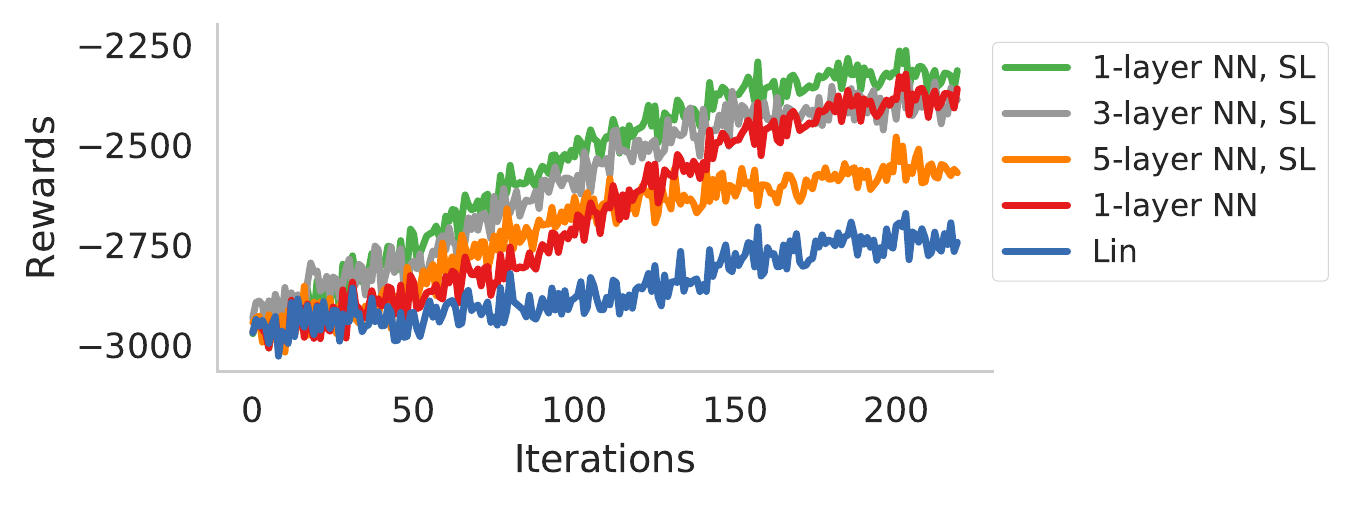}
  \caption{Learning performance of different control structures (1-layer neural network, 5-layer neural network, and linear controller). By the inclusion of a smoothness loss (SL) in the learning objective \eqref{equ:pol_obj}, the exploration becomes more effective.
  }
  \label{fig:rew_iters_pvtol}
\end{figure}

The learned 5-layer neural network policy is employed in an actual control task, as shown in \cref{fig:pvtol_state_act}. Compared to the nominal controller, the flights can be maneuvered more efficiently in this case with only local information. In terms of the actual cost, the RL agents achieve the cost 41.0, which is about 30\% lower than that of the nominal controller (58.3). This result can be examined both in the actual state-action trajectories in \cref{fig:pvtol_state_act} or the control behaviors in \cref{fig:pvtol_rl_nom}. The results indicate that RL is able to improve a given controller when the underlying system is nonlinear and unknown.

\begin{figure*}[!h]
  \centering
  \begin{subfigure}[h]{.32\columnwidth}
  \centering
        \includegraphics[width=\textwidth,trim=0mm 0mm 0mm 0mm,clip]{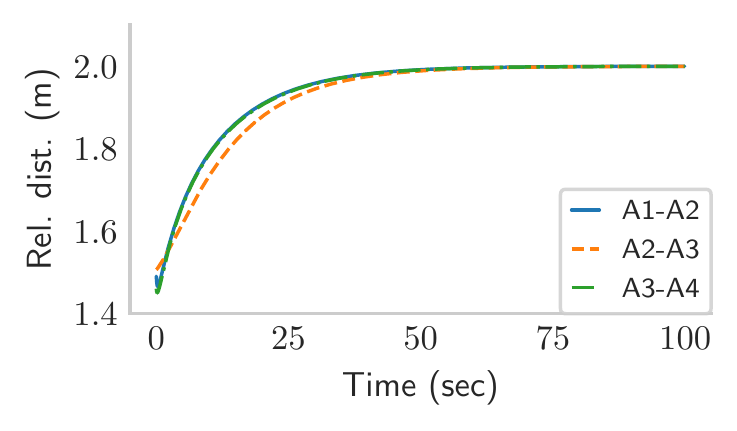}
        \caption{Nom.: reletive distances.}
        \label{fig:pvtol_x_nn_traj_nom}
    \end{subfigure}
      \begin{subfigure}[h]{.32\columnwidth}
  \centering
        \includegraphics[width=\textwidth,trim=0mm 0mm 0mm 0mm,clip]{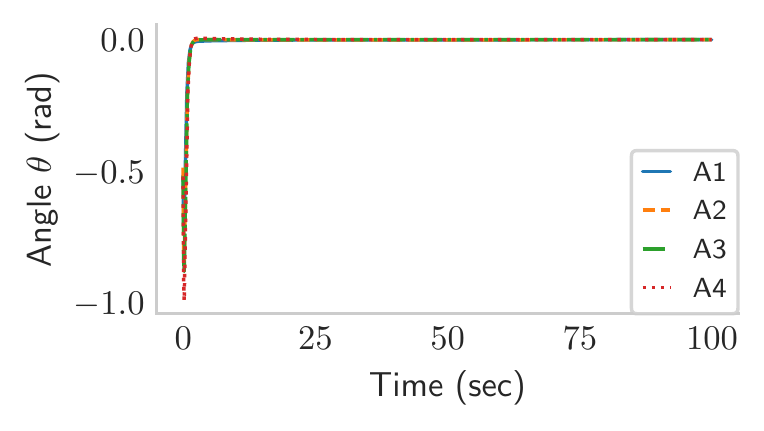}
        \caption{Nom.: angles $\theta_i$.}
        \label{fig:pvtol_angle_nn_traj_nom}
    \end{subfigure}
  \begin{subfigure}[h]{0.32\textwidth}
        \includegraphics[width=\textwidth]{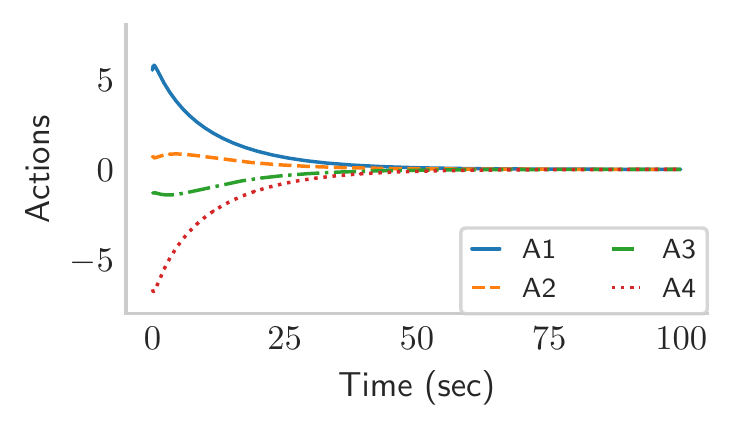}
        \caption{Nom.: actions.}
        \label{fig:pvtol_act_nn_traj_nom}
    \end{subfigure}
  \begin{subfigure}[h]{.32\columnwidth}
  \centering
        \includegraphics[width=\textwidth,trim=0mm 0mm 0mm 0mm,clip]{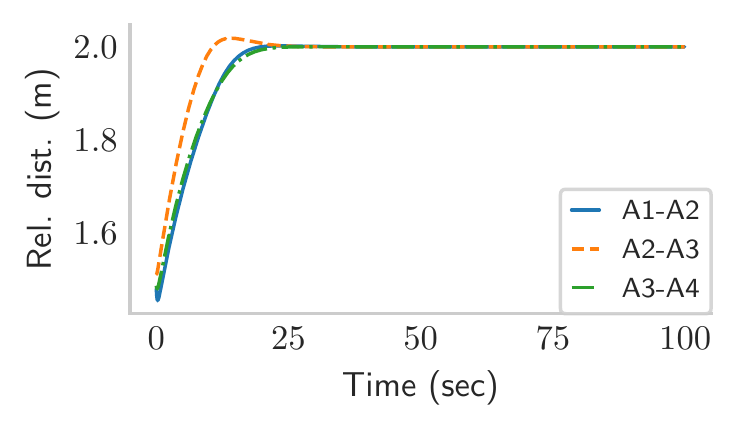}
        \caption{NN: relative distances.}
        \label{fig:pvtol_x_nn_traj}
    \end{subfigure}
      \begin{subfigure}[h]{.32\columnwidth}
  \centering
        \includegraphics[width=\textwidth,trim=0mm 0mm 0mm 0mm,clip]{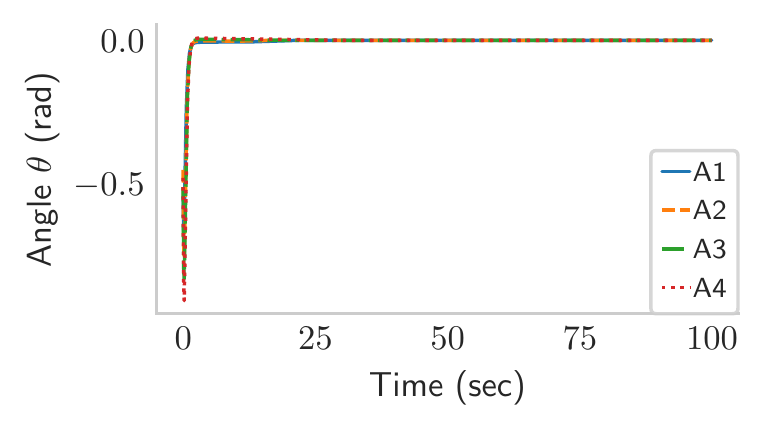}
        \caption{NN: angles $\theta_i$.}
        \label{fig:pvtol_angle_nn_traj}
    \end{subfigure}
  \begin{subfigure}[h]{0.32\textwidth}
        \includegraphics[width=\textwidth]{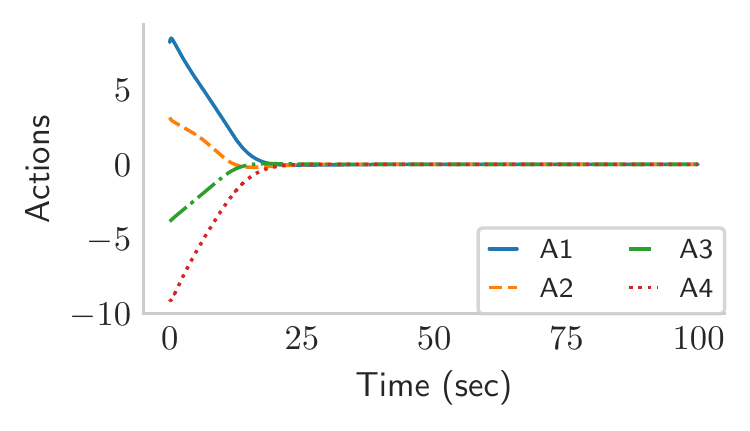}
        \caption{NN: actions.}
        \label{fig:pvtol_act_nn_traj}
    \end{subfigure}
  \caption{State and action trajectories in a typical contral task, where the nominal controller (Nom) and the RL agents achieve costs of 58.3 and 41.0, respectively.}
  \label{fig:pvtol_state_act}
\end{figure*}

\begin{figure*}[!h]
  \centering
  \begin{subfigure}[h]{.32\columnwidth}
  \centering
        \includegraphics[width=\textwidth,trim=0mm 0mm 0mm 0mm,clip]{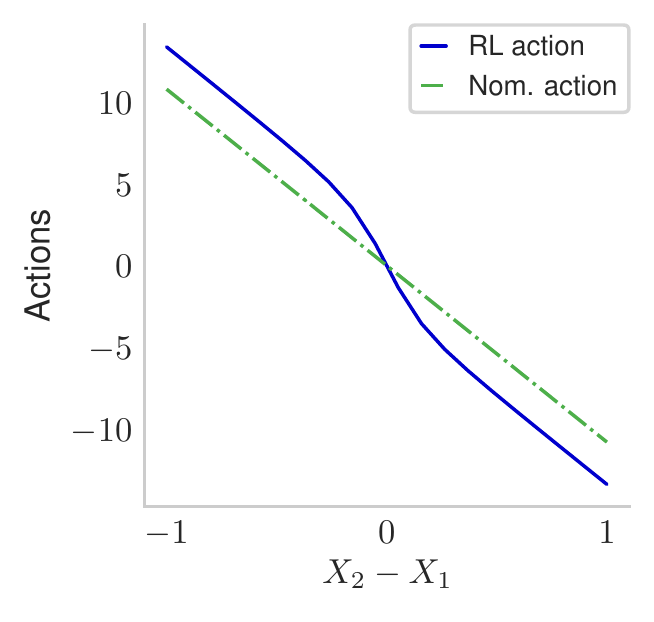}
        \caption{Controller of Agent 1.}
        \label{fig:pvtol_a1}
    \end{subfigure}
      \begin{subfigure}[h]{.32\columnwidth}
  \centering
        \includegraphics[width=\textwidth,trim=0mm 0mm 0mm 0mm,clip]{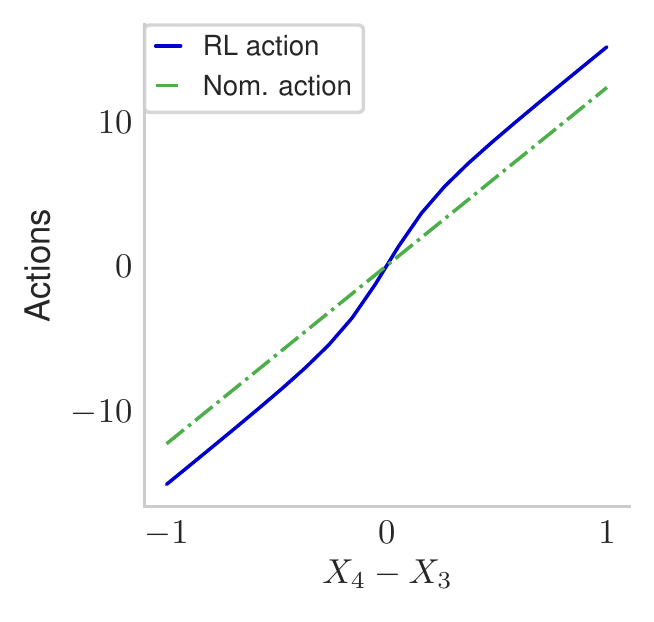}
        \caption{Controller of Agent 4.}
        \label{fig:pvtol_a4}
    \end{subfigure}
  \caption{Demonstration of control outputs for the nominal action and RL agents.}
  \label{fig:pvtol_rl_nom}
\end{figure*}

\subsection{Power system frequency regulation}

In this case study, we focus on the problem of distributed control for power system frequency regulation \cite{fazelnia2017convex}. The IEEE 39-Bus New England Power System under analysis is shown in \cref{fig:power_illustr}. In a distributed control setting, each generator can only share its rotor angle and frequency information with a pre-specified set of counterparts that are geographically distributed. The main goal is to optimally adjust the mechanical power input to each generator such that the phase and frequency at each bus can be restored to their nominal values after a possible perturbation. Let $\theta_i$ denote the voltage angle at a generator bus $i$ (in rad). The physics of power systems are modeled by the per-unit swing equation:
\begin{equation}
m_i\ddot{\theta}_i +d_i\dot{\theta}=p_{m_i}-p_{e_i}
\end{equation}
where $p_{m_i}$ is the mechanical power input to the generator at bus $i$ (in p.u.), $p_{e_i}$ is the electrical active power injection at bus $i$ (in p.u.), $m_i$ is the inertia coefficient of the generator at bus $i$ (in p.u.-sec$^2$/rad), and $d_i$ is the damping coefficient of the generator at bus $i$ (in p.u.-sec/rad). The electrical real power injection $p_{e_i}$ depends on the voltage angle difference in a nonlinear way, as governed by the AC power flow equation:
\begin{equation}
p_{e_i} = \sum_{j=1}^n |v_i||v_j|\left(g_{ij}\cos(\theta_i-\theta_j)+b_{ij}\sin(\theta_i-\theta_j)\right)
\end{equation}
where $n$ is the number of buses in the system, $g_{ij}$ and $b_{ij}$ are the conductance and susceptance of the transmission line that connects buses $i$ and $j$,  $v_i$ is the voltage phasor at bus $i$, and  $|v_i|$ is its voltage magnitude. Because the conductance $g_{ij}$ is typically several orders of magnitude smaller than the susceptance $b_{ij}$, for the simplicity of mathematical treatment, we omit the cosine term and only keep the sine term that accounts for the majority of nonlinearity. Each generator needs to make decisions on the value of the mechanical power $p_{m_i}$ to inject in order to maintain the stability of the power system.

\begin{figure*}[th]
  \centering
  \begin{subfigure}[h]{.46\columnwidth}
  \centering
        \includegraphics[width=\textwidth,trim=0mm 0mm 0mm 0mm,clip]{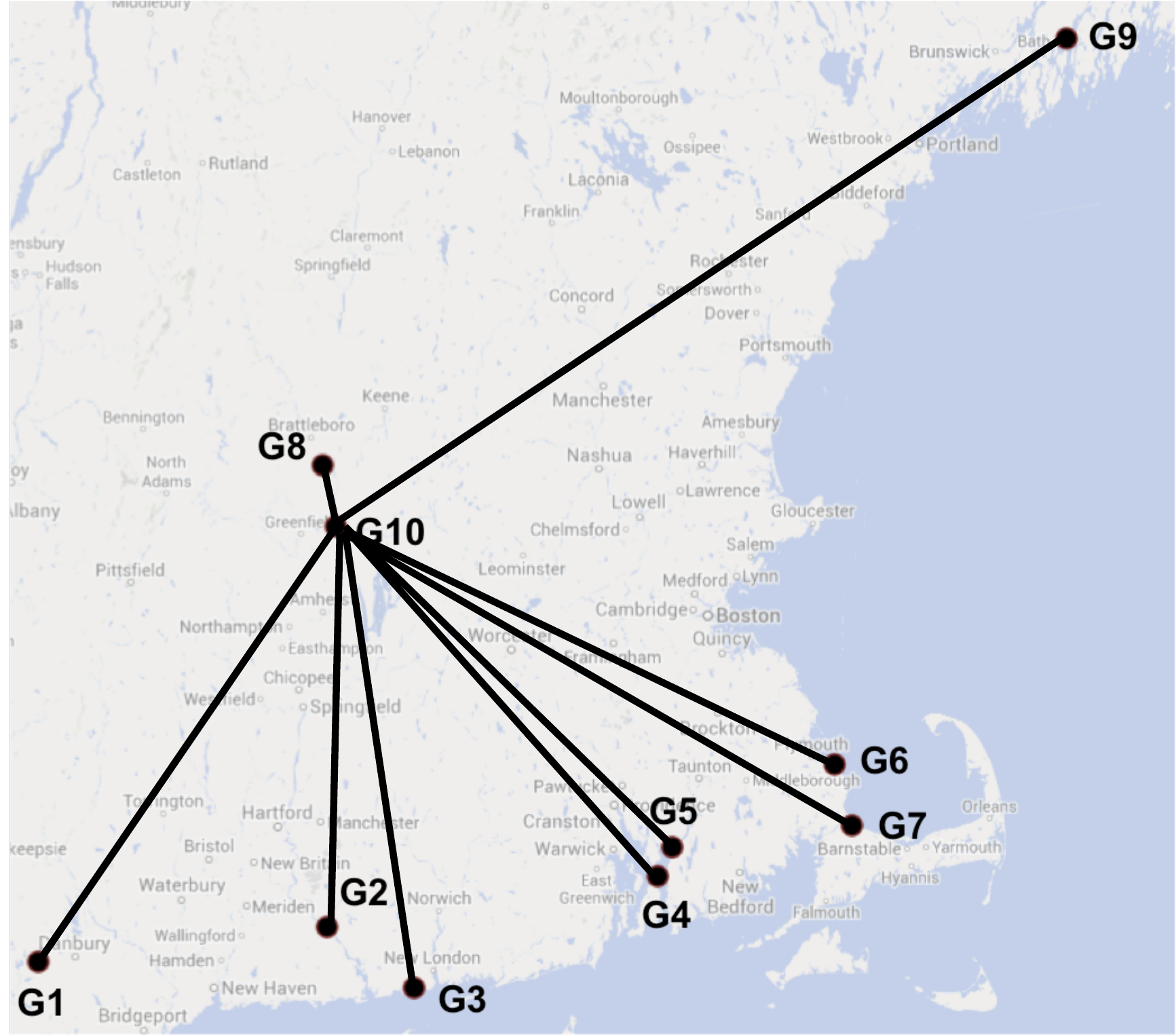}
        \caption{Star-connected structure.}
        \label{fig:star_power}
    \end{subfigure}
      \begin{subfigure}[h]{.4\columnwidth}
  \centering
        \includegraphics[width=\textwidth,trim=0mm 0mm 0mm 0mm,clip]{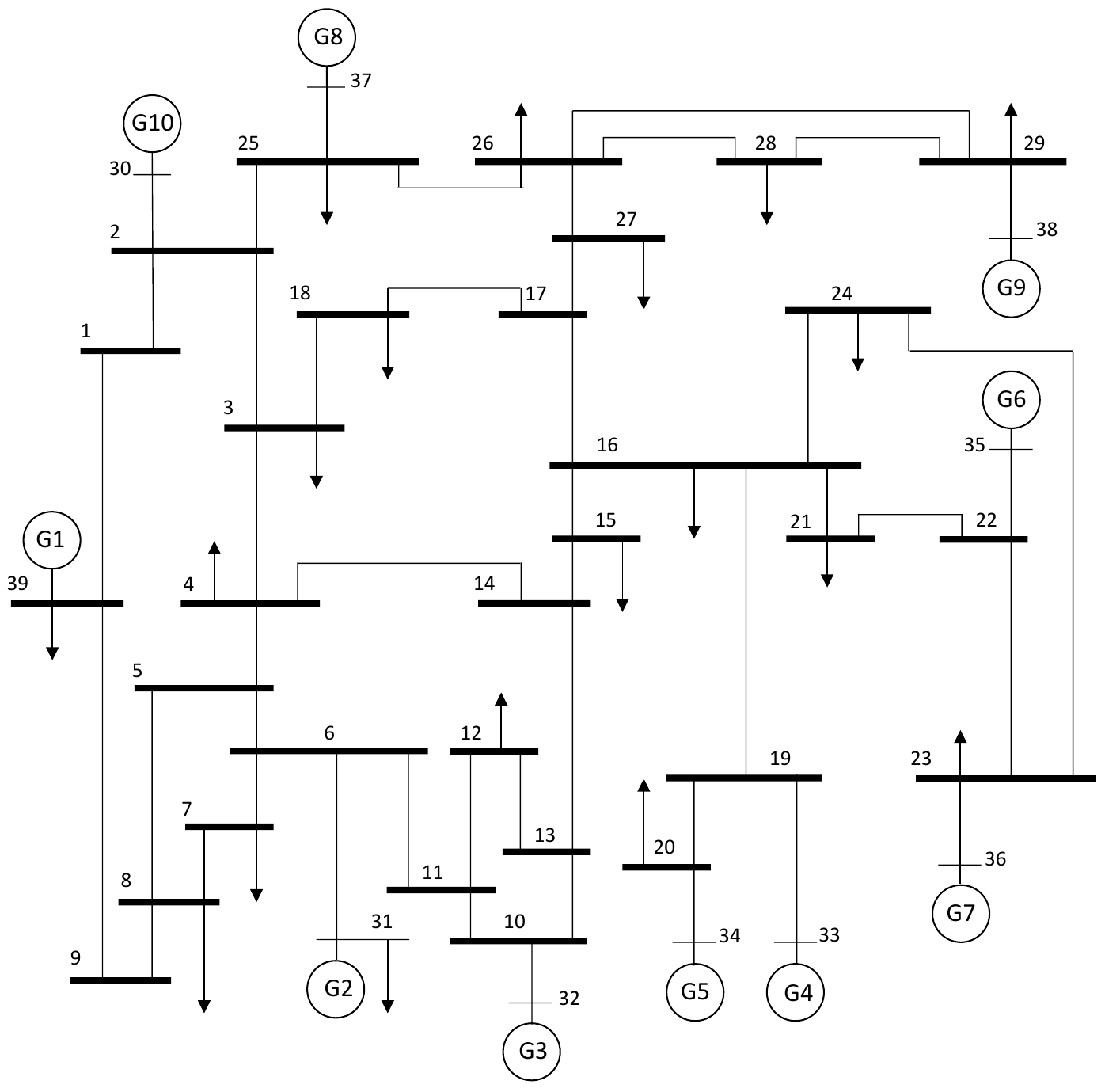}
        \caption{IEEE 39-bus system.}
        \label{fig:star_power_bus}
    \end{subfigure}
    \caption{Illustration of the frequency regulation problem for the New England power system. The communication among generators follows a star topology. }
    \label{fig:power_illustr}
\end{figure*}

Let the rotor angles and the frequency states be denoted as $\theta=\begin{bmatrix}
\theta_1&\cdots&\theta_n
\end{bmatrix}^\top$ and $\omega=\begin{bmatrix}
\omega_1&\cdots&\omega_n
\end{bmatrix}^\top$, and the generator mechanical power injections be denoted as $p_m=\begin{bmatrix}
p_{m_1}&\cdots&p_{m_n}
\end{bmatrix}^\top$. Then, the state-space representation of the nonlinear system is given by:
\begin{equation}
\begin{bmatrix}
\dot{\theta}\\
\dot{\omega}
\end{bmatrix}=\underbrace{\begin{bmatrix}
0&I\\
-M^{-1}L&-M^{-1}D
\end{bmatrix}}_{A}\underbrace{\begin{bmatrix}
{\theta}\\
{\omega}
\end{bmatrix}}_{x}+\underbrace{\begin{bmatrix}
0\\
M^{-1}
\end{bmatrix}}_{B}p_m+\underbrace{\begin{bmatrix}
\bZero\\
{g}(\theta)
\end{bmatrix}}_{g(x)}
\end{equation}
where ${g}(\theta)=\begin{bmatrix}
g_1(\theta)&\cdots&g_n(\theta)
\end{bmatrix}^\top$ with $g_i(\theta)=\sum_{j=1}^n \frac{b_{ij}}{m_j}\left((\theta_i-\theta_j)-\sin(\theta_i-\theta_j)\right)$, and $M=\text{diag}\left(\{m_i\}_{i=1}^n\right)$,  $D=\text{diag}\left(\{d_i\}_{i=1}^n\right)$, and $L$ is a Laplacian matrix whose entries are specified in \cite[Sec. IV-B]{fazelnia2017convex}. For linearization (also known as DC approximation), the nonlinear part $g(x)$ is assumed to be zero when the phase differences are small \cite{fazelnia2017convex,fattahi2017transformation}. On the contrary, we deal with this term in the stability certification to demonstrate its capability of producing non-conservative results even for nonlinear systems. Similar to the flight formation case, we assume that there exists a distributed nominal controller that stablizes the system. To conduct multi-agent RL, each controller $p_{m_i}$ is a neural network that takes the available phases and frequencies as the input and determines the mechanical power injection at bus $i$. The main focus is to study the {certified-gradient bounds} for each agent policy in this large-scale setting.

\textbf{Stability certificate:} Similar to the flight formation problem, the nonlinearities in $\bgv(\bx)$ are in the form of $\Delta\theta_{ij}-\sin\Delta\theta_{ij}$, where $\Delta\theta_{ij}=\theta_i-\theta_j$ represents the phase difference, which has its slope restricted to $[0,1-\cos(\overline{\theta})]$ for every $\Delta\theta_{ij}\in[-\overline{\theta},\overline{\theta}]$ and thus can be treated using the Zames-Falb IQC. In the smoothness margin analysis, assume that $\overline{\theta}=\frac{\pi}{3}$, which requires the phase angle difference to be within $[-\frac{\pi}{3},\frac{\pi}{3}]$. This is a large set of uncertainties that includes both normal and abnormal operational conditions. To study the stability of the multi-agent policies, we adopt a black-box approach by simply considering the input-output constraint. By simply applying the $\Log_2$ constraint in \eqref{equ:ineq_const_1}, we can only certify stability for Lipschitz constants up to 0.4, as shown in \cref{fig:lip_gains_power}. Because the distributed control is sparse, we can leverage it by setting the lower and upper bounds $\underline{\xi}_{ij}=\overline{\xi}_{ij}=0$ for each agent $i$ that does not utilize observation $j$, and $\overline{\xi}_{ij}=-\underline{\xi}_{ij}=l$ otherwise, where $l$ is the Lipschitz constant to be certified. This information can be encoded in $\underline{\text{SDP}}(P,\lambda,\gamma,\xi)$ in \eqref{equ:nltv_sdp}, which can be solved for $L$ up to 0.6 (doubling the certificate provided by the $\Log_2$ constraint).

\begin{figure}[h!]
  \centering
  \includegraphics[width=0.7\textwidth]{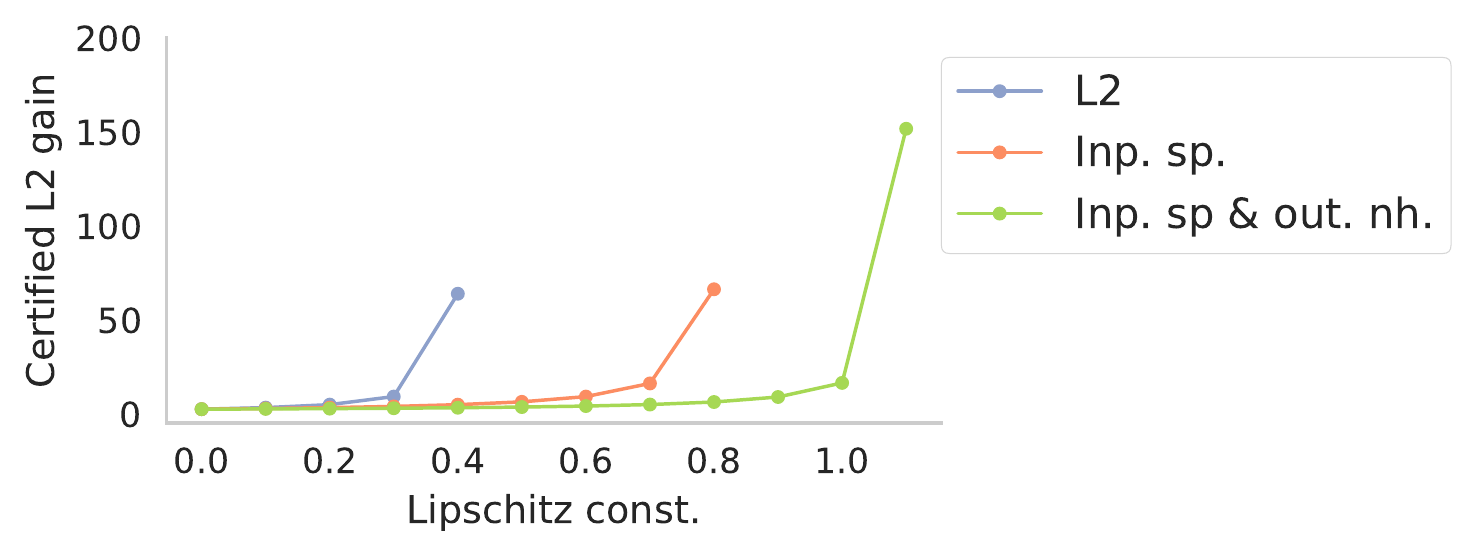}
  \caption{Certified Lipschitz constants for the power system regulation task.
  }
  \label{fig:lip_gains_power}
\end{figure}

Due to the problem nature, we further observe that for each agent, the partial gradient of the policy with respect to certain observations is primarily one-sided, as shown in \cref{fig:power_grads}. With a band of $\pm 0.1$, the partial gradients remain within either $[-0.1,1]$ or $[-1,0.1]$ throughout the learning process. This information is gleaned during the learning phase, and we can incorporate it into the partial gradient bounds (e.g., $\overline{\xi}_{ij}=-0.1l$ and $\underline{\xi}_{ij}=l$ for agent $i$ which exhibits positive gradient with respect to observation $j$) to extend the certificate up to 1.1.

\begin{figure*}[!h]
  \centering
  \begin{subfigure}[h]{.54\columnwidth}
  \centering
        \includegraphics[width=\textwidth,trim=0mm 0mm 0mm 0mm,clip]{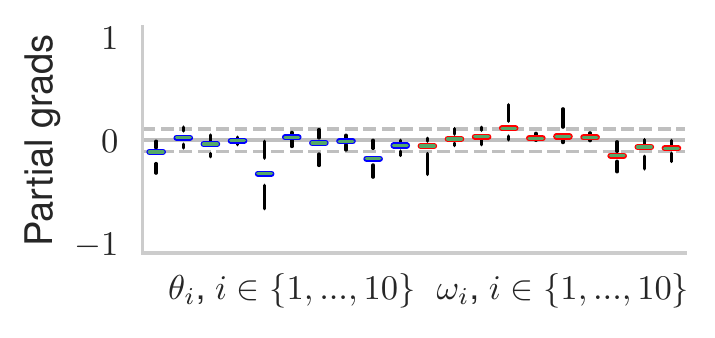}
        \caption{G10.}
        \label{fig:cen_ag_grad}
    \end{subfigure}
      \begin{subfigure}[h]{.215\columnwidth}
  \centering
        \includegraphics[width=\textwidth,trim=8mm 0mm 0mm 0mm,clip]{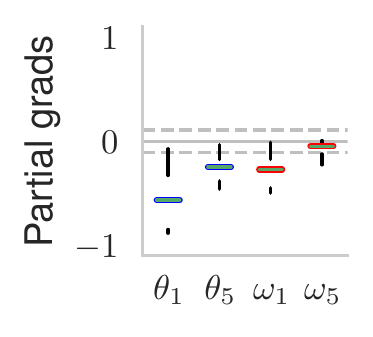}
        \caption{G4.}
        \label{fig:ag4_grad}
    \end{subfigure}
  \begin{subfigure}[h]{0.215\textwidth}
        \includegraphics[width=\textwidth,trim=8mm 0mm 0mm 0mm,clip]{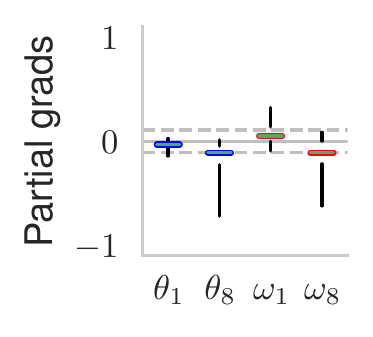}
        \caption{G7.}
        \label{fig:ag7_grad}
    \end{subfigure}
  \caption{Box plots of partial gradients of individual generators (G10, G4, G7) with respect to local information. Grey dashed lines indicate $\pm 0.1$.}
  \label{fig:power_grads}
\end{figure*}

\begin{figure*}[!h]
  \centering
  \begin{subfigure}[h]{.42\columnwidth}
  \centering
        \includegraphics[width=\textwidth,trim=0mm 0mm 0mm 0mm,clip]{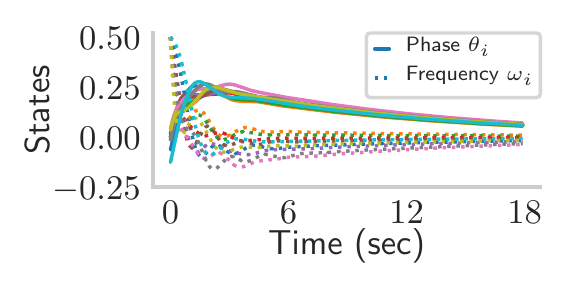}
        \caption{Nom.: phase $\theta_i$ and frequency $\omega_i$.}
        \label{fig:phase_dnn3_traj_nom}
    \end{subfigure}
  \begin{subfigure}[h]{0.57\textwidth}
        \includegraphics[width=\textwidth]{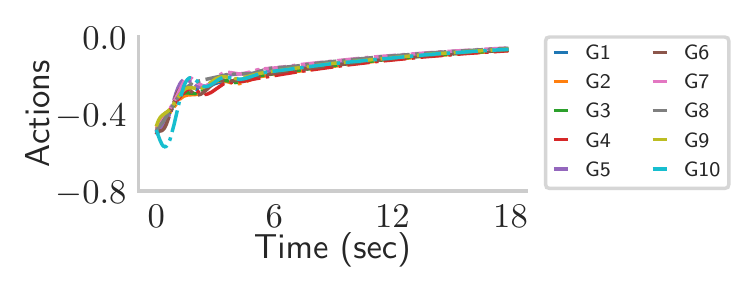}
        \caption{Nom.: actions.}
        \label{fig:act_dnn3_traj_nom}
    \end{subfigure}
     \begin{subfigure}[h]{.42\columnwidth}
  \centering
        \includegraphics[width=\textwidth,trim=0mm 0mm 0mm 0mm,clip]{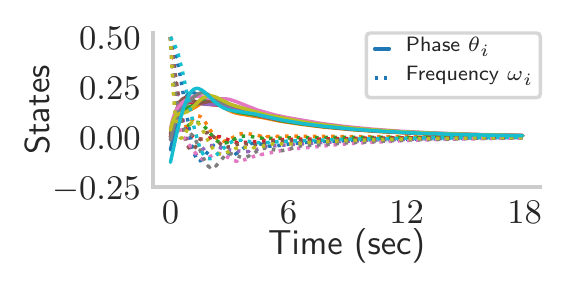}
        \caption{NN: phase $\theta_i$ and frequency $\omega_i$.}
        \label{fig:phase_dnn3_traj_nnn}
    \end{subfigure}
  \begin{subfigure}[h]{0.57\textwidth}
        \includegraphics[width=\textwidth]{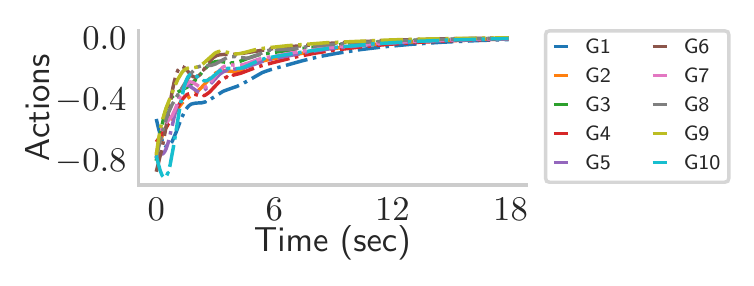}
        \caption{NN: actions.}
        \label{fig:act_dnn3_traj_nnn}
    \end{subfigure}
  \caption{State and action trajectories of the nominal and neural network controllers for power system frequency regulation, with costs of 50.8 and 23.9, respectively.}
  \label{fig:power_state_act}
\end{figure*}

\textbf{Policy gradient RL:} Similar to the flight formation task, we perform multi-agent policy gradient RL. The learned neural network controller is implemented in a typical control case, whose trajectories are shown in \cref{fig:power_state_act}. As can be seen, the RL policies can regulate the frequencies more efficiently than the nominal controller, with a significantly lowered cost (50.8 vs. 23.9). More importantly, we compare the cases of RL with and without regulating the Lipschitz constants in \cref{fig:rew_iters_power_long}. Without regulating the gradients, the RL is able to reach a performance slightly higher than its stability-certified counterpart. However, after about iteration 500, the performance starts to deteriorate (due to a possibly large gradient variance and high sensitivity to step size) until it completely loses the previous gains and starts to break the system. This intolerable behavior is due to the large Lipschitz gains that grow unboundedly, as shown in \cref{fig:lip_traj_power}. In comparison, RL with regulated gradient bounds is able to make a substantial improvement, and also exhibits a more stable behavior.

\begin{figure}[h!]
  \centering
  \includegraphics[width=0.7\textwidth]{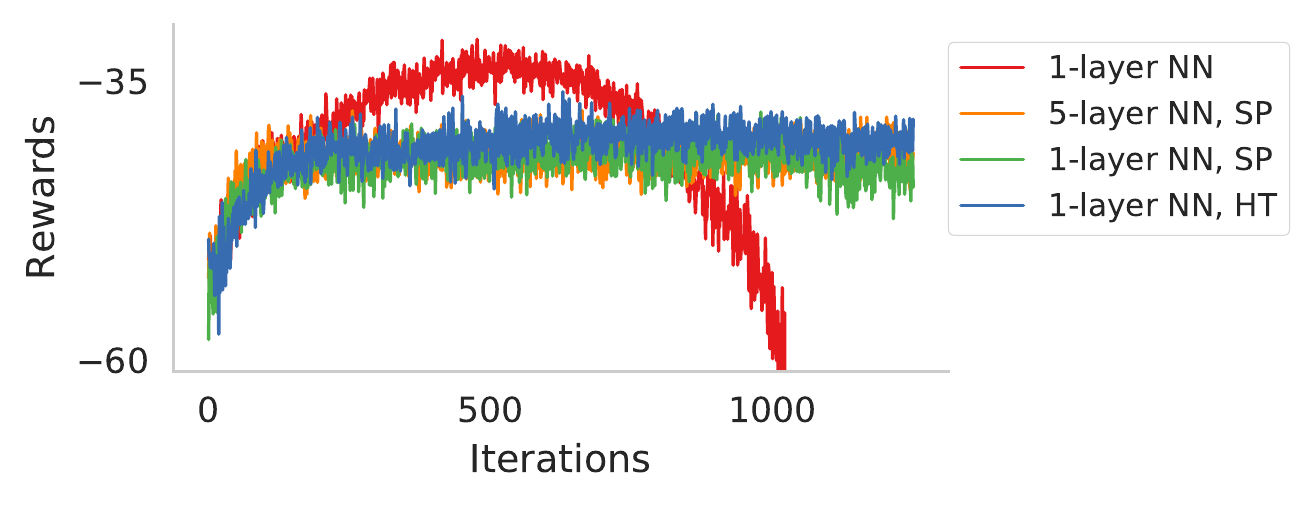}
  \caption{Long-term performance of RL for agents with regulated gradients by soft penalty (SP), which adaptively adjusts the coefficients $\omega_2$ in \eqref{equ:pol_obj}, and hard thresholding (HT), which shrinks the network last layer to satisfy the gradient bounds. The RL agents without regulating the gradients exhibit ``dangerous'' behaviors in the long run.
  }
  \label{fig:rew_iters_power_long}
\end{figure}

\begin{figure*}[h!]
  \centering
  \begin{subfigure}[h]{.4\columnwidth}
  \centering
        \includegraphics[width=\textwidth,trim=0mm 0mm 30mm 0mm,clip]{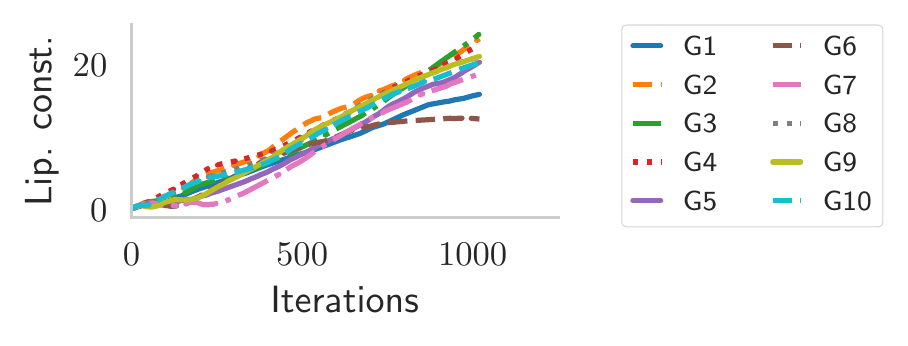}
        \caption{No gradient regulation.}
        \label{fig:lip_traj_power_nn}
    \end{subfigure}
  \begin{subfigure}[h]{0.58\textwidth}
        \includegraphics[width=\textwidth]{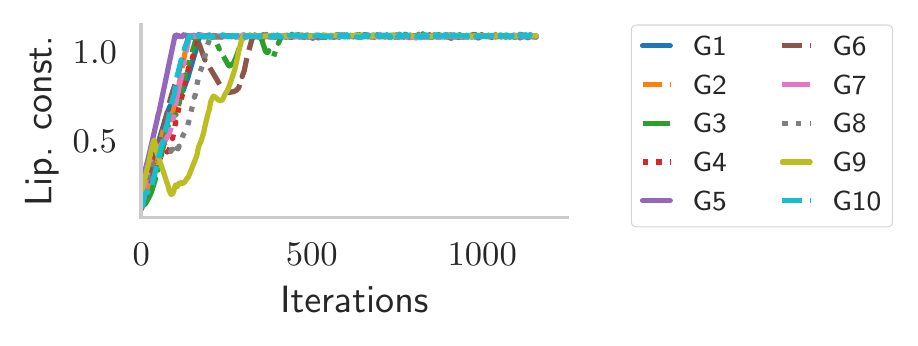}
        \caption{Gradient regulation.}
        \label{fig:lip_traj_power_nn_ht}
    \end{subfigure}
  \caption{Trajectories of Lipschitz constants with and without regulation.
  }
  \label{fig:lip_traj_power}
\end{figure*}

\section{Conclusions}
\label{sec:conclusions}

In this paper, we focused on the challenging task of ensuring the stability of reinforcement learning  in real-world dynamical systems. By solving the proposed SDP feasibility problem, we can offer a preventative certificate of stability for a broad class of neural network controllers with bounded gradients. Furthermore, we analyzed the (non)conservatism of the certificate, which was demonstrated in the empirical investigation of  {decentralized nonlinear control} tasks, including multi-agent flight formation and power grid frequency regulation. Results indicated that the set of stability-certified controllers was significantly larger than what the existing approaches can offer, and that the RL agents can substantially improve the performance of nominal controllers while staying within the safe set. Most importantly, regulation of gradient bounds was able to improve on-policy learning stability and avoid ``catastropic'' effects caused by the unregulated high gains. The present study represents a key step towards safe deployment of reinforcement learning in mission-critical real-world systems.


\bibliographystyle{siamplain}
\bibliography{rob_rl}
\end{document}